\documentclass[12pt, leqno]{article}

\usepackage[margin=1.0in]{geometry} 

\usepackage{float}

\usepackage{amsmath,amsthm,amssymb,}
\usepackage{setspace}
\usepackage[dvipsnames]{xcolor}
\usepackage{wasysym, marvosym} 

\usepackage{color}
\usepackage{stmaryrd}
\usepackage{enumerate}

\usepackage{tabularx}
\usepackage{booktabs}
\usepackage{enumitem}

\usepackage{nicefrac}
\usepackage{multicol}
\usepackage{amsthm}
  \usepackage[pdftex,
  bookmarks = false,
  pdfstartview = FitBH,
  linktocpage = true,
  pagebackref = true,
  pdfhighlight = /O,
  colorlinks=true,
  linkcolor=blue,
  citecolor=blue,
  filecolor = blue,
  urlcolor = blue,
  menucolor = blue,
]{hyperref}

\usepackage[colorinlistoftodos]{todonotes}

\usepackage{fancyhdr}
\usepackage{physics}
\usepackage{soul}
\usepackage{makecell}
 
 \usepackage{tikz} 
\usepackage{tikzscale}
\usetikzlibrary{positioning,arrows}
   \tikzset{
   modal/.style={>=stealth,shorten >=1pt,shorten <=1pt,auto,node distance=1.5cm,
   semithick},
   world/.style={circle,draw,minimum size=0.5cm,fill=gray!15},
   point/.style={circle,draw,inner sep=0.5mm,fill=black},
   reflexive above/.style={->,loop,looseness=7,in=110,out=70},
   reflexive below/.style={->,loop,looseness=7,in=240,out=300},
   reflexive left/.style={->,loop,looseness=7,in=150,out=210},
   reflexive right/.style={->,loop,looseness=7,in=30,out=330}}
\usetikzlibrary{calc, shapes}

\usetikzlibrary{shadows,arrows,positioning}
 \usepackage{tikz-network}
 
\usepackage{tabularx}
\usepackage{multirow}

\usepackage{tablefootnote} 

\newcommand{\David}[2][]{\todo[color=cyan, #1]{#2}}
\newcommand{\Chris}[2][]{\todo[color=violet, #1]{#2}}
\newcommand{\Marian}[2][]{\todo[color=green, #1]{#2}}

\newtheorem{theorem}{Theorem}

\newtheorem{axiom}{Axiom}
\usepackage{mathtools}

\usepackage[normalem]{ulem}

\usepackage{comment}
\usepackage[symbol]{footmisc}

\title{How Haag-tied is QFT, really?}

\author{
\textsc{Chris Mitsch}\footnote{\href{mailto:chris.mitsch@pitt.edu}{chris.mitsch@pitt.edu}}\\ 
\textsc{Marian Gilton}\footnote{\href{mailto:marian.gilton@pitt.edu}{marian.gilton@pitt.edu}}\\
\textsc{David Freeborn}\footnote{\href{mailto:dfreebor@uci.edu}{dfreebor@uci.edu} }\\
}

\date{\today}

\begin{document}
\renewcommand{\arraystretch}{1.25}

\maketitle

\begin{abstract}
    Haag's theorem cries out for explanation and critical assessment: it sounds the alarm that something is (perhaps) not right in one of the standard way of constructing interacting fields to be used in generating predictions for scattering experiments.  Viewpoints as to the precise nature of the problem, the appropriate solution, and subsequently-called-for developments in areas of physics, mathematics, and philosophy differ widely. In this paper, we develop and deploy a conceptual framework for critically assessing these disparate responses to Haag's theorem. Doing so reveals the driving force of more general questions as to the nature and purpose of foundational work in physics. 
\end{abstract}
\thispagestyle{fancy}

\section{Introduction}

\begin{quote}
    The edifice of science is not raised like a dwelling, in which the foundations are first firmly laid and only then one proceeds to construct and to enlarge the rooms. Science prefers to secure as soon as possible comfortable spaces to wander around and only subsequently, when signs appear here and there that the loose foundations are not able to sustain the expansion of the rooms, it sets about supporting and fortifying them. This is not a weakness, but rather the right and healthy path of development.
\end{quote}
[Hilbert 1905, 102; translation by Corry 2004, 127]\\

\pagestyle{plain}

Proven over six decades ago, Haag's theorem appears to present a problem for particle physics. The theorem seems to block a key technique\textemdash namely, the interaction picture and its attendant calculational methods\textemdash that has been widely used to generate successful predictions. It is clear that the theorem points to some sort of problem, driven by the empirical success of the calculations employing the interaction picture on the one hand and the logical force of the theorem on the other.  
Thus, while particle physics has secured for itself a ``comfortable space'' around which to wander, Haag's theorem appears as a sign that the foundation are too loose ``to sustain the expansion of the rooms.''

This paper aims to provide a framework for possible answers to a single, if double-faced, question: \textit{What does Haag's theorem tell us about quantum field theory, present and future?} Several divergent answers have been given already. Indeed, these will shape the paper's framework substantially. Nevertheless, the shape the framework should take is less straightforward than it may at first seem. Even before getting to the nitty-gritty analysis of Haag's theorem, a framework must grapple with the problem of viewpoint, as the following exercise makes clear:

Regardless of your actual field, take whatever career stage you are at\textemdash early graduate student, doctoral candidate, early-, middle- or late-career researcher\textemdash and imagine yourself instead as a particle physicist. You may imagine yourself as an experimentalist or a theoretician, expert in QED or QCD\textemdash whatever comes to mind. Regardless, several things are true of you. First, you are committed to the development of particle physics (no matter what this means in practice). Second, you are steeped in, and reliant on, the interaction picture for your research in and teaching. And third, you have just learned of Haag's theorem and the trouble that it spells for the interaction picture. How might you respond?

Set the exercise up again, except now you are a mathematician committed to contributing to the  development of QFT (no matter what this means in practice). Second, you are steeped in the implications of Haag's theorem, and you are intimately familiar with the axiomatic or algebraic approach to quantum field theory. And third, you believe that a full, conceptually coherent physically realistic replacement for the interaction picture is (presently) unavailable.

Set it up one last time, except now you are a philosopher of science. You may consider yourself a realist or an instrumentalist, interested in metaphysics or methodology\textemdash whatever comes to mind. Regardless, several things are true of you. First, you are committed to understanding the foundations of QFT (no matter what this means in practice). Second, you are familiar both with the major advances in canonical (Lagrangian) QFT that use the interaction picture as well as those based on algebraic QFT (AQFT).\footnote{See \cite{fraser2011take} and \cite{wallace2011taking} for the classic debate over these two formulations of QFT.} But third, it is unclear to you if, or how, these advances can form a consistent whole.

It is not a given that your physicist, mathematician, and philosopher selves will share a single outlook on RQFT, even before considering Haag's theorem. Nor is it clear that they should. Inevitably, this will affect your reaction to our guiding question. Thus, a major contribution of our framework will be to highlight the influence these extra-Haagian outlooks have on understanding Haag's theorem and on assessing its implications for philosophy and for theoretical physics. As we show in section \ref{sec:framework}, many (but not all) of the disagreements about Haag's theorem derive ultimately from different extra-Haagian outlooks, such that the disagreement is far less about Haag's theorem itself than it is about how to do (foundations of) physics. 

This work also aims to contribute to a larger discussion in philosophy of physics. As a quintessential and live example of work at the foundations of physics, the discussions of Haag's theorem 
 draw our attention to important methodological questions: What role does (should) foundational work play in progress in physics? How is foundational work coordinated with non-foundational work, or how should it be? And, what does (should) foundational work even look like? These are undoubtedly heady questions, and we feign no complete answers. Nevertheless, our framework will reveal some of the answers that are being given by the authors under survey, and in so doing these answers open themselves to investigation. As we conclude the paper, we suggest several investigations we expect will sober up future discussions of these heady methodological questions.

The remainder of the paper is structured as follows. In section \ref{sec:HaagThm} we offer a cursory historical overview, a synopsis of a standard proof of Haag's theorem, and discuss the sense in which it raises an alarm that something is not right with the interaction picture. In section \ref{sec:lit} we motivate the need for a framework for understanding the literature on Haag's theorem. The framework itself is given in section \ref{sec:framework}. The framework employs and extends Hilbert's construction analogy: the framework understands each author as something like a contractor giving their \textbf{assessment} of the problem in the foundations of physics heralded by Haag's theorem, their recommendation for \textbf{repair} work on the foundations, and their expectations of the needed long-term \textbf{maintenance or future renovations} to QFT's area of the edifice of science. This section applies the framework to seven leading contemporary viewpoints on Haag's theorem; the key results of this application are given in table \ref{tab:summary}. Section \ref{sec:diagrams} demonstrates  the framework's judicious balance of conceptual structure and flexibility in order to bring clarity to the space of responses to Haag's theorem; it further argues that, at the end of the day, the most important lesson for philosophers to take from Haag's theorem is that we need to put our own energies into clearly answering  meta-level questions regarding the nature and purpose of foundational work in physics. Concluding remarks are given in section \ref{sec:conclusion}.  

We will use the following terms to disambiguate the different approaches to and versions of QFT. By \textit{canonical QFT} we mean what \cite{wallace2001emergence} does by \textit{Lagrangian QFT}: an approach that, for the most part, proceeds from specification of a classical field theory, to its quantization, and finally to its (Wilsonian) renormalization. By \textit{formal variants} of QFT we mean the most familiar collection of approaches focused on mathematical rigor and precision, including algebraic QFT, constructive QFT (CQFT), and axiomatic QFT; these center around the Wightman, Osterwalder-Schrader, or Haag-Kastler axioms.

\section{Haag's Theorem}\label{sec:HaagThm}

\subsection{History of Haag's Theorem}\label{subsec:history}

Haag's theorem is the culmination of two approaches in early quantum field theory. On the one hand, Dyson had combined Feynman's rules and the Tomonaga-Schwinger formalism into a reliable and practical approach for calculating the results of scattering experiments in relativistic quantum field theory. This approach brought with it a shift toward considering scattering amplitudes between \textit{free, asymptotic states} (see \cite{blum2017thestate}), there accomplished by the so-called interaction picture to model interactions \cite{schwinger1948quantum} (see \ref{sec:IP}). The approach was wildly successful, in particular in its use in calculating the anomalous magnetic moment of the electron \cite{Schwinger_magmomel}. This is the genesis of canonical QFT.

On the other hand, spurred by Wigner's precise characterization of special relativity's implications for quantum theory \cite{wigner1939unitary}, a more mathematically rigorous approach to understanding the structure of relativistic quantum field theories maintained an interest in determining the properties of \textit{instantaneous states} in addition to free, asymptotic states. This was the genesis of formal variants QFT. Naturally, a transparent representation of relativistic QFT's mathematical structure was prized, leading to an enumeration of several assumptions. However, the assumptions underlying the former approach's use of the interaction picture were unclear. Haag, in his lecture series given at CERN in 1952--3,\footnote{See \cite{haag2010some} for his recollections of this aspect of the history. Note that Haag says (incorrectly) that the lectures were given in 1953--4; see \cite{lupher2005proved}.} began by stating assumptions for models of relativistic QFT for a single, fundamental field. These included, following Wigner's suggestion, requiring an irreducible Hilbert space representation of the Poincar{\'e} group. Then, Haag used the mathematical tools at his disposal to characterize the interaction picture (e.g., the particle representation of field states, unitary intertwiner). The problem (Haag's theorem \cite{haag1955quantum}) was that the interaction picture, thus characterized, could not be used to represent non-trivial interactions. This result was clarified and generalized by \cite{hall1957theorem}, to which most discussions of Haag's theorem refer. Essentially, Haag, Hall, and Wightman showed that ``to give different physical predictions, two theories of an interacting field which satisfies the canonical commutation relations must use inequivalent representations of the commutation relations'' \cite[2]{hall1957theorem}.

We should make two historical remarks. First, Haag's theorem was not unique in observing the existence and importance of unitarily inequivalent representations of the CCRs. \cite[329]{reedsimon_ii} suggest---and \cite[21]{haag1955quantum} asserts---that von Neumann may have known of the existence of inequivalent representations by the time he wrote \cite{vonNeumann1938oninfinite}, and, von Neumann notwithstanding, the particle representation originating in \cite{dirac1927quantum} was the only one used until the late 1940s. But by 1951, Friedrichs had characterized a new representation and demonstrated its general inequivalence with the particle representation. He moreover argued that ``in a certain sense [such fields] \textit{do occur in nature} \cite[142;153]{friedrichs1953aspects}, and, in fact, the representation allows one to treat problems previously insoluble because of the so-called infrared catastrophe. Similarly, \cite{wightman1955configuration} produced inequivalent representations (using von Neumann's techniques) and suggest their physical relevance; this built, in part, on a model van Hove used to demonstrate that interacting and free state spaces can be orthogonal \cite{vanhove1952difficultes}.\footnote{See \S7.2.2 of \cite{Sbisa:2020whw} for an explanation of how van Hove's result complements, but does not follow from, Haag's theorem. We also note with Sbis{\`a} that van Hove's result implicates ultraviolet divergences, unlike those of Friedrichs or Haag (as Friedrichs himself noted \cite[270-72]{friedrichs1953aspects}).} Garding and Wightman then published their classification of all representations of the CARs in \cite{GardingWightman1954CCRs}, noting along the way that unitary equivalence also depends on the ``terminal behavior'' of the involved unitary operators when the measure (partially) characterizing the representation is nondiscrete \cite[621]{GardingWightman1954CCRs}. This line of thinking led to results like Theorem X.46 of \cite[233]{reedsimon_ii}, which says that free, neutral scalar fields of inequivalent mass carry unitarily inequivalent irreducible representations. This theorem will be relevant in \S4, where e.g., \cite{duncan2012conceptual,klaczynski2016haags,EarmanFraser2006implications}, construe this theorem as the conceptual core of Haag's theorem.

Second, the early reception of Haag's theorem focused on its implications for a particle interpretation but neither \textit{reduced} its significance to its interpretive implications nor viewed it as ultimately damning for particles. It will suffice to discuss Haag himself. While he emphasized its consequences for the ``possibility of defining a theory[\dots] which describes the interaction processes of particles,'' he nevertheless (a) notes in summary that relaxing the equal time vanishing of the CCRs suffices to evade the theorem and (b) claims the assumptions he makes are not only compatible but even ``have physical significance'' in the lowest orders of a perturbation expansion \cite[36-7]{haag1955quantum}. Haag actually emphasizes two escape paths taken by later commentators, namely that his use of strong convergence diverges from the typical convergence factor and that satisfactory physical descriptions do not require use of infinite dimensional Hilbert spaces. The point for Haag, it seems, was to probe the limits of the conventional field theory scheme \cite[32]{haag1955quantum}.

 \subsection{The Interaction Picture}
 \label{sec:IP}
 

Haag's theorem is often construed as a no-go theorem for the interaction picture. As previously noted, the interaction picture is a mainstay of undergraduate and graduate textbooks, and has facilitated the calculation of many physical quantities that have matched experimental results to a high degree of accuracy, \footnote{\label{note:GellMannLow}, for example the celebrated computation of the anomalous magnetic dipole moment of the electron by \cite{Schwinger_magmomel}. As one crude measure of how salient the interaction picture is for actual physics, a search for the exact phrase ``interaction picture'' in Physics Reviews conducted on Nov 4, 2021 returned 4,377 hits. The interaction picture forms the basis of calculations using time-dependent perturbation theory. Crucially, the Gell-Mann and Low theorem makes use of the interaction picture. The original paper \cite{Gell-Mann_Low_1951} presenting this formula has (as of February 4, 2024) 1637 citations recorded on google scholar. This is likely a rather low estimate of how widely the formula is in fact used since the formula is now at the level of common knowledge among physicists, and thus in many cases it might be used without citation. For instance, the Gell-Mann and Low theorem was used to rigorously derive the Bethe–Salpeter equation, which describes the bound states of a system of two fermions  in a relativistic formalism \cite{Bethe_Salpeter_1951}. This equation was used in subsequent quantum electrodynamics calculations of the fine structure of Hydrogen-like atoms \cite{Salpeter_1952}, Helium \cite{Douglas_Kroll_1974} and heavy atoms \cite{Mohr_Plunien_Soff_1998} and contemporary precision calculations of the electron magnetic moment \cite{Kinoshita_2015}. Likewise, the equations have been utilized in the calculation of nucleon-nucleon potentials \cite{Lomon_Partovi_1969, Lomon_1972, Lomon_1976, Lomon_Partovi_1978} and the derivation of scaling laws for interactions with large momentum transfer, confirmed by scattering experiments \cite{Brodsky_Farrar_1975}, among many other examples.} 


As its name suggests, the interaction picture is one way to model interacting fields in canonical quantum field theory. In the Schr{\"o}dinger picture of quantum mechanics, states evolve in time under the full Hamiltonian, whilst operators are stationary. In the Heisenberg picture of quantum mechanics, the operators evolve under the full Hamiltonian, whilst states are stationary. The interaction picture is intermediate between these two pictures. To form the interaction picture, we split the full Hamiltonian into a free and a (time-dependent) interaction part, $H= H_F + H_I$. The evolution of operators is governed by the free Hamiltonian, $H_F$, so the fields are free. The evolution of states is governed by the interaction part, $H_I$.

We stipulate that the operator fields coincide with those of the Heisenberg picture at some time, $t_0$. Let $V(t_2, t_1)$ represent the unitary evolution of the interaction picture states from time $t_1$ to $t_2$, generated by the interacting part of the Hamiltonian, $H_I$,

\begin{equation}
V(t_2, t_1) = e^{-i H_I (t_2 - t_1)} = e^{+i H_F (t_2 - t_1)}e^{-i H (t_2 - t_1)}.
\end{equation}

\noindent We call this operator the intertwiner, or Dyson operator. Then, at all times, $t$, the Heisenberg (subscript $F$) and interaction (subscript $I$) operators are related by the intertwiner as follows,

\begin{align}
\phi_I (t, \mathbf{x}) &= V^{-1}(t,t_0) \phi_H (t, \mathbf{x}) V(t,t_0), \\
\pi_I (t, \mathbf{x}) &= V^{-1}(t,t_0) \pi_H (t, \mathbf{x}) V(t,t_0),
\end{align}

\noindent where at time $t_0$ the operators interaction picture and Heisenberg operators coincide, $\phi_I (t_0, \mathbf{x}) = \phi_H (t_0, \mathbf{x})$ and $\pi_I (t_0, \mathbf{x}) = \pi_H (t_0, \mathbf{x})$. As we will see, it is central to Haag's theorem that the relation of the interaction field, $\phi_I$, to the free field, $\phi_F$ is characterized by a unitary map. 

Generally, we seek to calculate physical amplitudes, taken in the limit in which the fields are free at times $t \rightarrow \pm \infty$. This is meant to capture the intuition that in an interaction, particles begin infinitely far apart (and hence not interacting) and then separate again infinitely far apart after an interaction. The interaction picture may be useful if we can treat the effects of $H_I$ as a small, time-dependent perturbation on the evolution under $H_F$. In perturbation theory we perform approximate calculations by expanding the desired physical quantities in powers of the small interaction, $H_I$.

\subsection{Proof of Haag's Theorem for Spin-free, Neutral, Scalar Fields}

Haag's theorem follows from the results of two other theorems. For clarity of exposition, here we only prove the theorem for neutral, scalar fields; the generalization for other types of fields follows in a straightforward manner. The core strategy is due to \cite{hall1957theorem} rather than \cite{haag1955quantum}. The main technical advance of the former over the latter was a series of theorems concerning the analytic continuation of functions invariant under the orthochronous inhomogeneous Lorentz group, which allowed them to extend the equality of two-, three-, and four-point vacuum expectation values at equal times to any times. For simplicity, we eschew reference to this complex, noting its intuitive gloss as ``a quantitative formulation of the intuitive feeling that in a Lorentz-invariant theory the equivalence of descriptions in different Lorentz frames should somehow restrict the possible correlations between the values of physical quantities at different points in space-time'' \cite[35]{hall1957theorem}. We likewise appeal to the simpler Jost-Schroer theorem, which at any rate affords extension to all vacuum expectation values (when one of the fields is free). For convenience, we also adopt some of the notation and conventions of \cite{EarmanFraser2006implications} and \cite{seidewitz2017avoiding}, including the numbering due to the latter.

\subsubsection{Wightman Axioms}
\setcounter{axiom}{-1}

\begin{axiom}\textbf{States.} We have a physical Hilbert Space, $\mathcal{H}$, for which the states, $\ket{\phi}$, are rays, such that,
\begin{enumerate}[label=0.\arabic*]
\item The states transform according to a continuous unitary representation of the Poincar\'e group, $U(\mathbf{\Delta x}, \Lambda)$, under Poincar\'e transformations, $ \{ \mathbf{\Delta x}, \Lambda\}$. \label{axiom:unitaryPoincare}
\item There is a unique, invariant vaccuum state, $\ket{0}$, in $\mathcal{H}$, invariant under $U$: $U(\mathbf{\Delta x}, \Lambda) \ket{0} = \ket{0}$.
\label{axiom:uniquevacuum}
\item Let $U(\mathbf{\Delta x}, \Lambda) = e^{i P^\mu \mathbf{\Delta x}_\mu}$. Then, $P^\mu P_\mu = - m^2$. We interpret $P^\mu$ as an energy-momentum operator and $m$ as a mass. The eigenvalues of $P^\mu$ lie in the future lightcone.
\end{enumerate}
\end{axiom}

\begin{axiom}\textbf{Domain and continuity of fields.} The field $\phi(x)$ and its adjoint $\phi^\dagger (x)$ are defined on a domain $D$ of states dense in $\mathcal{H}$ containing the vacuum state $\ket{0}$. The $U(\mathbf{\Delta x}, \Lambda)$, $\phi(x)$ and $\phi^\dagger (x)$ all transform vectors in $D$ to vectors in $D$.
\end{axiom}

\begin{axiom}\textbf{Field transformation law.} The fields transform under Poincar\'e transformations as,
\begin{equation}
U(\mathbf{\Delta x}, \Lambda) \phi(x) U^{-1}(\mathbf{\Delta x}, \Lambda) = \phi (\Lambda x + \mathbf{\Delta x}).
\end{equation}
\end{axiom}

\begin{axiom}\textbf{Local commutativity.} \label{axiom:localcommutativity} If $x$ and $x'$ are two spacetime positions,
\begin{equation}
[\phi(x), \phi(x') ] = [\phi ^\dagger (x), \phi ^\dagger (x')] = 0
\end{equation}
Furthermore, if $x$ and $x'$ are space-like separated,
\begin{equation}
[\phi(x), \phi ^\dagger (x') ] = 0.
\end{equation}
\end{axiom}

\begin{axiom}\textbf{Cyclicity of the vacuum.} The vacuum state $\ket{0}$ is cyclic for the fields, $\phi(x)$. That is, polynomials in the fields and their adjoints, when applied to the vacuum state, yield a set $D_0$ dense in $\mathcal{H}$.
\end{axiom}

\subsubsection{Proof of Haag's theorem}

From here on, it is convenient to decompose four-vectors as $x = (t, \mathbf{x})$, where the right hand side are the temporal and spatial components, respectively.

\begin{theorem}\textbf{Equality of equal-time vacuum expectation values.}
\label{EETVEVs} 
Let $\phi_1$ and $\phi_2$ be two field operators, with associated conjugate momentum operators, $\pi_1$ and $\pi_2$, defined in respective Hilbert spaces $\mathcal{H}_1$ and $\mathcal{H}_2$, satisfying the Wightman axioms listed above, and satisfying the equal time commutation relations,

\begin{align}
\label{eq:CCRs}
[\phi_i(t,\mathbf{x}), \pi_i(t, \mathbf{x}')] &= i\delta(\mathbf{x}-\mathbf{x}'), \\
[\phi_i(t,\mathbf{x}), \phi_i(t, \mathbf{x}')] &= [\pi_i(t,\mathbf{x}), \pi_i(t,\mathbf{x}')] = 0.
\end{align}

Suppose that there exists a unitary operator $G$ such that, at some specific time $t_0$, 

\begin{align}
\phi_2 (t_0,\mathbf{x}) &= G \phi_1 (t_0,\mathbf{x}) G^{-1}, 
\label{eqn:unitaryrelationbetween2fields_1} \\
\pi_2 (t_0,\mathbf{x}) &=  G \pi_1 (t_0,\mathbf{x}) G^{-1}.
\label{eqn:unitaryrelationbetween2fields_2} 
\end{align}


\noindent We call G an intertwiner for the fields $\phi_1$ and $\phi_2$. Then the equal-time vacuum expectation values of the fields coincide. (Note that so far, this holds only at the particular time, $t_0$.) 
\end{theorem}

\begin{proof}

Let $U_i(\mathbf{\Delta x}, \mathbf{R})$ be a continuous, unitary representation of the inhomogeneous Euclidean group of translations, $\mathbf{\Delta x}$, and three-dimensional rotations $\mathbf{R}$, defined on each $\mathcal{H}_i$, $i=1,2$. Let  us further suppose that  transformations $U_i(\mathbf{\Delta x}, \mathbf{R})$ induce Euclidean transformations of the field

\begin{align}
U_i(\mathbf{\Delta x}, \mathbf{R}) \phi_i (t_0,\mathbf{x}) U_i ^{-1} (\mathbf{\Delta x}, \mathbf{R}) &= \phi (t_0, \mathbf{R}\mathbf{x} + \mathbf{\Delta x}), \\
U_i(\mathbf{\Delta x}, \mathbf{R}) \pi_i (t_0,\mathbf{x}) U_i ^{-1} (\mathbf{\Delta x}, \mathbf{R})  &= \pi (t_0, \mathbf{R}\mathbf{x} + \mathbf{\Delta x}), 
\end{align}

\noindent  as in axiom \ref{axiom:unitaryPoincare}. From our supposition (equations \ref{eqn:unitaryrelationbetween2fields_1} and \ref{eqn:unitaryrelationbetween2fields_2}), it follows that

\begin{equation}
U_2 (\mathbf{\Delta x}, \mathbf{R}) = G U_1 (\mathbf{\Delta x}, \mathbf{R}) G^{-1}.
\end{equation}

\noindent But since the representations possess unique invariant vacuum states $\ket{0}_i$ such that  $U_i(\mathbf{\Delta x}, \Lambda) \ket{0}_i = \ket{0}_i$, as in axiom \ref{axiom:uniquevacuum}.,

\begin{equation}
c \ket{0}_2 = G \ket{0}_1,
\end{equation}

\noindent where $c$ is a complex number of absolute value 1, $|c| = 1$. In other words, up to a phase factor, $G \ket{0}_1$ is the vacuum state for field $\phi_2$ at time $t_0$.

It further follows that the equal-time vacuum expectation values (Wightman functions, also known as correlation functions) of the two fields are the same,

\begin{equation}
   \prescript{}{1}{ \bra{0}} \phi_1 (t_0, \mathbf{x}_1), \ldots \phi_1 (t_0, \mathbf{x}_n) \ket{0}_1 
= \prescript{}{2}{\bra{0}} \phi_2 (t_0, \mathbf{x}_1), \ldots \phi_2 (t_0, \mathbf{x}_n) \ket{0}_2 ,
\end{equation}

\noindent for $\mathbf{x}_1, \mathbf{x}_2$ up to at most $\mathbf{x}_4$, all at the same fixed time, $t_0$.\footnote{If at least one of $\phi_1$ or $\phi_2$ is a free field, then this can be proven to hold for all Wightman functions, i.e. for $\mathbf{x}_1, \mathbf{x}_2$ up to any $\mathbf{x}_n$.}
\end{proof}

\begin{theorem}\textbf{Jost-Schroer theorem.} For any free scalar field $\phi$, the two-point vacuum expectation values are given by

\begin{equation}
 \bra{0} \phi(x_1) \phi^\dagger (x_2) \ket{0} = \Delta^+ (x_1-x_2),
\label{eqn:2pointvacuum}
\end{equation}

\noindent where $ x_1 = (t_1, \mathbf{x}_1) $ and $ x_2 = (t_2, \mathbf{x}_2) $ are arbitrary spacetime four-vectors. $\Delta^+$ is the advanced Feynman propagator\footnote{For example, see \cite[pages 209-210 and 256]{duncan2012conceptual}}, given by, 
\begin{equation}
\Delta^+ (x_1 - x_2) = (2\pi)^{-3} \int d^3p \frac{e^{i [-\omega_{\mathbf{p}} (t_1 - t_2) + \mathbf{p} \cdot (\mathbf{x}_1 -\mathbf{x}_2)]}}{2\omega_{\mathbf{p}}},
\end{equation}
and $\omega_p = \sqrt{\mathbf{p}^2 + m^2}$. 

If, for any arbitrary field, the vacuum expectation values are given by equation \ref{eqn:2pointvacuum}, then that field is a free field.

\end{theorem}

This result was first proved by Jost \cite{Jost_1961} for fields of positive mass, and extended to fields of zero mass by Pohlmeyer \cite{Pohlmeyer_1969}.

\begin{theorem}\textbf{Haag's theorem for scalar fields.} Let $\phi_1$ be a free, scalar field, which therefore satisfies equation \ref{eqn:2pointvacuum}. Let $\phi_2$ be a second, locally Lorentz-covariant scalar field. Let us assume that $\phi_2$ is unitarily related to $\phi_1$ at time $t_0$, as in equations \ref{eqn:unitaryrelationbetween2fields_1} and \ref{eqn:unitaryrelationbetween2fields_2}. 
Further let us assume that the conjugate momenta fields (written as adjoints $\phi_{1}^\dagger$ and $\phi_{2}^\dagger$) satisfy the hypotheses of Theorem \ref{EETVEVs}. Then $\phi_2(x)$ is also a free field.

\end{theorem}

\begin{proof}
Theorem \ref{EETVEVs} tells us that the equal time vacuum expectation values of the two fields, $\phi_1$ and $\phi_2$ must coincide at $t_0$, i.e.,
\begin{equation}
   \prescript{}{1}{ \bra{0}} \phi_1 (t_0, \mathbf{x}_1) \phi_1 (t_0, \mathbf{x}_2) \ket{0}_1 
= \prescript{}{2}{\bra{0}} \phi_2 (t_0, \mathbf{x}_1) \phi_2 (t_0, \mathbf{x}_2) \ket{0}_2 .
\end{equation}

\noindent Since the field $\phi_1$ is free, it follows from Theorem 2 (equation \ref{eqn:2pointvacuum}) and equation 20 that the two-point vacuum expectation values for the two fields coincide at $t_0$, i.e.,

\begin{equation}
\prescript{}{2}{\bra{0}} \phi_2 (t_0, \mathbf{x}_1) \phi_2^\dagger (t_0, \mathbf{x}_2) \ket{0}_2 = \prescript{}{1}{\bra{0}} \phi_1 (t_0, \mathbf{x}_1) \phi_1^\dagger (t_0, \mathbf{x}_2) \ket{0}_1.
\label{eqn:2pointvacuumforphi2_t0}
\end{equation}

So far, this holds \textit{only} at $t_0$. However, any two spacelike separated position vectors, $(t_1, x_1)$ and $(t_2, x_2)$, can be brought into the equal time plane $t_1 = t_2$ by a Lorentz transformation. 
Thus, the Lorentz-covariance of $\phi_2$ allows us to extend the satisfaction of equation \ref{eqn:2pointvacuumforphi2_t0} to any two spacelike positions, and then, because time translation is equivalent to a Lorentz boost plus spatial translation plus Lorentz boost \cite[317]{EarmanFraser2006implications}, to any two positions:

\begin{equation}
\prescript{}{2}{\bra{0}} \phi_2 (x_1) \phi_{2}^\dagger (x_2) \ket{0}_2 = \Delta^+ (x_1-x_2).
\label{eqn:2pointvacuumforphi2}
\end{equation}

\noindent Therefore, by Theorem 2, $\phi_2(x)$ must be a free field.

\end{proof}

\subsection{Implications of Haag's Theorem}
\label{sec:implications-haag}


By the criteria used by particle physicists, the interaction picture has undoubtedly produced numerous notable successes (see note \ref{note:GellMannLow}). However, the interaction picture is generally taken to depend upon all of the assumptions needed for Haag's theorem, including Poincar\'e invariance and the existence of a unitary operator relating the two fields, that is, that the fields are unitarily equivalent. The apparent clash between the interaction picture and Haag's theorem arises as follows. If the fields, $\phi_F$ and $\phi_I$, obey the Wightman axioms of QFT, and we have a unitary operator intertwining the two fields at even a single time (as in equations \ref{eqn:unitaryrelationbetween2fields_1} and \ref{eqn:unitaryrelationbetween2fields_2}), 
 and $\phi_F$ is free, then according to Haag's theorem $\phi_I$ must also be free. So the interaction must be trivial.

So Haag's Theorem appears to be a no-go theorem for calculations that use the interaction picture. At a glance, it would be mathematically inconsistent to use the interaction picture for calculations involving any of the non-trivial interactions that we care about in particle physics. However, a closer look reveals a whole labyrinth of philosophical, mathematical, and physical issues at stake in understanding the full significance of Haag's theorem. We take that closer look in the next section. 


\section{What the Haag Is Going on?}\label{sec:lit}
There are a few  points on which all parties generally agree. First, judged by its own particular standards of success, particle physics --including the interaction picture-- is highly successful. 

Second, there is a consensus that Haag's theorem poses a \textit{bona fide} problem for the standard presentation of the interaction picture.\footnote{More precisely, it is agreed to pose a problem for the pre-renormalization interaction picture.} The theorem itself is mathematically correct: as Klaczynski puts it, ``it is a mathematical theorem in the truest sense of the word; it brings with it the `hardness of the logical must'.'' \cite{klaczynski2016haags}. Further, it is not disputed that the assumptions of Haag's theorem hold in the case of the standard textbook presentations of the interaction picture.\footnote{For example, standard presentations assume, either implicitly or explicitly, a Poincar\'e-invariant theory, and that there are distinct free and interacting fields, related by a unitary operator.} No one, to our knowledge, argues that the problem posed by Haag's theorem is simply illusory.\footnote{It may very well be that many practicing physicists today would say that the problem, if not illusory, is inconsequential. Perhaps many physicists think that the problem posed by Haag's theorem is at most rooted in an unrealistic idealization of non-interacting states at temporal infinity. And perhaps they therefore choose not to address Haag's theorem in their textbooks, lecture notes, or research articles. Although several textbooks from the  1960's give brief assessments of the theorem, more recent treatments rarely mention it (see \cite[307-309]{EarmanFraser2006implications}). We can do little more than speculate that such an unspoken consensus fully explains the dearth of references to Haag's theorem in standard accounts of QFT\textemdash simple ignorance of the theorem may be just as strong of a causal factor. David Tong's lecture notes on QFT, for instance, do not explicitly mention Haag's theorem; but they do say that the assumption of non-interacting states in the interaction picture is wrong and should be replaced by the interactions-always-on interpretation of the LSZ reduction formula \cite[54-55, 79-80]{tong2006qftnotes}. We might tentatively read this as an attitude towards Haag's theorem as minimally useful: it points out that an initial assumption was wrong, and in wise pedagogical fashion, Tong will correct that assumption in due course when more complexity and sophistication can be digested. But no need to belabor the point by naming and deeply discussing the theorem. Thus, while it is possible that this sort of largely dismissive response to Haag's theorem is widespread, too little of it exists in print to be extensively covered in the remainder of this paper.} It is, rather, the severity and appropriate remedy of the problem that is subject to debate.

However, the severity of the problem raised by Haag's theorem clearly stops short of spelling the demise of that research program predicated on the application of quantum theories of fields to scattering experiments, commonly called particle physics.\footnote{Though see the discussion of Kastner (section \ref{subsec:Kastner}) and Seidewitz (section \ref{subsec:Seidewitz}) below for genuine proposals of new physical theories, each drawing some motivation from Haag's theorem. } 
Even the practitioners of AQFT, operating far afield from the details of the Standard Model's phenomenology, still conceive of their work as contributing to the scientific enterprise whose primary goal is to develop the quantum theory of fields as the appropriate theoretical apparatus for understanding scattering experiments.  So, then, what \textit{is} going on with Haag's theorem such that these two stances \textemdash the interaction picture has been used successfully, and Haag's theorem poses a substantive problem for the interaction picture\textemdash can be held together? 
This central tension 
is widely recognized as a point of agreement. As Teller succinctly puts it, 

\begin{quotation}
\noindent
Everyone must agree that as a piece of mathematics Haag's theorem is a valid result that at least appears to call into question the mathematical foundation of interacting quantum field theory, and agree that at the same time the theory has proved astonishingly successful in application to experimental results. \cite[115]{teller2020interpretive}
\end{quotation}

And yet, despite this initial agreement, extant responses to Haag's theorem form a confusing lot. The literature on Haag's theorem reflects a number of different assessments of the import of the theorem for both (mathematical) physics and philosophy. Moreover, the extent to which these different assessments make meaningful contact with each other is often unclear. In this section, we briefly illustrate the nature of the confusion in this literature.  First, the confusion is \textit{not} about the status of Haag's theorem as a mathematical result. All parties agree that the original theorem, and its several generalizations, have valid proofs. The confusion enters when trying to trace out the ramifications of Haag's theorem for the foundations of QFT. \cite{EarmanFraser2006implications} put it well, saying that ``the theorem provides an entry point into a labyrinth of issues that must be confronted in any satisfactory account of the foundations of QFT'' (p. 334). Once we have entered into this labyrinth via Haag's theorem, we encounter a host of conceptual and interpretive issues, enmeshed in technical issues of mathematics and physics, making even the range of options for a way forward through the labyrinth unclear, much less which one may be the best. 

A reader interested in Haag's theorem and its implications for the use of the interaction picture in physics may first look for insight from Earman's and D. Fraser's seminal paper, ``Haag's theorem and its implications for the foundations of quantum field theory.'' They conclude,  ``On any reading Haag's theorem undermines the interaction picture and the attendant approach to scattering theory'' \cite[333]{EarmanFraser2006implications}. So, the reader naturally thinks, \textit{the interaction picture is no good}. And yet for Duncan, ``the proper response to Haag's theorem is simply a frank admission that the same regularizations needed to make proper mathematical sense of the dynamics of an interacting field theory at each stage of a perturbative calculation will do double duty in restoring the applicability of the interaction picture at intermediate stages of the calculation'' \cite[370]{duncan2012conceptual}\textemdash the interaction picture survives! Miller concurs, adding that the success of calculations delivered from regularized and renormalized theories is explained by the conjecture that ``perturbative expansions are asymptotic to exact solutions of a theory that generates them'' \cite[818]{miller2018haag}.  \textit{So}, the reader concludes triumphantly, \textit{the interaction picture works}!

Still more paths begin to emerge, however. According to Klaczynski, these renormalized theories evade Haag's theorem precisely by \textit{denying} that the interaction picture exists \cite{klaczynski2016haags}. Maiezza and Vasquez agree, arguing that ``due to \textit{Haag's theorem}, it is \textit{impossible to define QFT} starting from the interaction picture with free fields'' \cite[10]{maiezza2020haag} (italics in original); indeed, they seem to argue that the interaction picture fails precisely because of the failure of the conjecture Miller relies on to save it. Yet confusingly, Maiezza and Vasquez \textit{also} disagree with Klaczynski on what saves perturbative calculations from Haag's theorem.

The paths so far, while many, nevertheless seem to turn on what to say about the mathematical coherence of the interaction picture. Thus, \textit{a labyrinth though it may be}, the reader thinks, \textit{I can at least see its basic structure}. The reader has judged too soon, however, for it is not only the interaction picture \textit{per se} that is at stake, but the metaphysics: ``\textit{either the assumptions of Haag’s theorem do not hold, in which case there is no particle notion applicable to a scattering experiment at intermediate times, or they do, in which case the particle notion applicable at intermediate times is incommensurable with the ingoing/outgoing particle notions, if the interaction is non-trivial}'' \cite[252]{ruetsche2011interpreting} (italics in original). \textit{So}, the reader thinks, \textit{the path out of the Haagian labyrinth requires the banishment of particles and the embracing of fields} (as \cite{halvorson2002no} suggest)! Not so, says Kastner: \textit{particles} exist, and \textit{fields} must be banished \cite{kastner_directaction}.

The reader is thus confronted with paths out of the Haagian labyrinth diverging on both metaphysical and mathematical grounds and, worse still, she cannot tell from her place in the labyrinth whether or where these paths coincide. As if this were not bad enough, a fog sets in. \textit{Are we even trapped at all?}, the reader asks, \cite[356]{seidewitz2017avoiding} in hand, for Haag's theorem arises in traditional QFT only because time is not ``treated comparably to the three space coordinates, rather than as an evolution parameter.'' Thus, the Haagian labyrinth could have been entirely avoided had time been treated in a relativistically sensible manner at the outset.

How should the reader react to this? Is there a Haagian labyrinth? And if so, which path will lead us out? The framework given in the next section gives a fixed structure for organizing and assessing this confusing network of responses to Haag's theorem, thereby creating several distinct mappings of the Haagian labyrinth. More general lessons from studying these maps are given in section \ref{sec:diagrams}.

\section{The Framework: Assessment, Repair, and Renovation and Maintenance}\label{sec:framework}

Before presenting the framework, we should carefully consider its goal. The purpose of this framework, as we said in the introduction, is to fruitfully structure and organize answers to the question ``What does Haag's theorem tell us about quantum field theory, present and future?''  Two desiderata for such a framework are fairly obvious. First, it should sensibly organize the answers that have already been given to this question. As such, we will apply the framework to prominent answers in the literature. 
Second, the framework should leave significant room for future developments. Given the interdisciplinary nature of the study of Haag's theorem, and given the recent increase of interest from physicists (\cite{klaczynski2016haags}, \cite{maiezza2020haag}, and \cite{seidewitz2017avoiding}), we expect there is much more to be said on this topic. A good framework for organization, therefore, must have space to accommodate these expected future developments.

Expanding on Hilbert's construction analogy, we will construe QFT as a building within the ``edifice of science'' that is continually under construction. Responses to Haag's theorem, then,  will be organized and analyzed according to their assessment or diagnosis of the problem posed by the theorem, the appropriate immediate repairs to that problem, and any longer-term renovations or maintenance measures called for in light of the assessment:
\begin{itemize}
    \item \textbf{Assessment:} What precisely is the problem posed by Haag's theorem, if any? For what objective(s) is this a problem? Or, are there multiple problems? 
    \item \textbf{Repair:} How should this problem be remediated?
    \item \textbf{Maintenance or Renovation:} Where should resources (time, attention, grant funding, conference and journal platforms, etc.) for the next (relevant) phase of research be allocated?
\end{itemize}
On one hand, the authors surveyed below may be thought of as subcontractors brought in to assess and diagnose the building's ability to support further growth in light of Haag's theorem. On the other hand, we the authors of this present paper, may be thought of as general contractors. As such we aim to arrange and make comparable these subcontractors' assessments, which we do in \ref{sec:framework}.


However, our ultimate goal as `general contractor' is to provide guidance on the building's construction\textemdash specifically, guidance as to how stakeholders in the future of QFT can coordinate their efforts on their shared interdisciplinary goals (see section \ref{sec:diagrams} below).
 This requires us to explain any disagreements amongst the subcontractors. The points of general agreement provide a convenient starting point. Recall from the outset of section \ref{sec:lit} that all authors agree:

\begin{enumerate}
    \item  particle physics is successful \textit{by its own standards}; and
    \item  Haag's theorem presents a \textit{bona fide} problem for using the interaction picture to relate distinct fields \textit{prior to renormalization}. 
\end{enumerate}
Differences of Assessment clarify where interlocutors in fact part ways beyond these points. Differences of Repair and Renovation reveal differences of motivation, expectation, and aim, rather than matters of fact. We have called these \textit{extra-Haagian outlooks} in order to stress that they are more abstract, give rise to less adversarial disagreements, and stem from commitments that extend far beyond the technical reach of Haag's theorem. Each of our two initial points of agreement gives rise to a distinctive locus of extra-Haagian outlooks.

Agreement (1) belies the different standards one might hold predictions to. It is agreed that the calculational methods produce reliable predictions. That is, the methods have generated successful predictions, and their application has been systematic enough to warrant belief in their continued success. Yet, insofar as they are generated using the interaction picture, these predictions rely on low-order perturbative approximations made at asymptotic times and relative to detector resolution. Moreover, since the full perturbative series diverges, it is difficult to specify what we are successfully approximating in a physical sense. \textit{Can we, or should we, expect more of a theory's predictions? of a relativistic QFT's predictions?}

Agreement (2) belies the different expectations one might have for the form of QFT. Beyond furnishing a method of generating predictions, one pre-Haagian appeal of the interaction picture was its intuitive representation of interaction dynamics. The asymptotic nature of calculations aside, interactions looked familiar: the full state space was time-evolved, from beginning to end, by a specified dynamical operator. As the agreement on (2) reflects, Haag's theorem undermines the interaction picture's ability to represent interactions in this way. And while renormalization delivers the agreed, reliable predictions of (1), it does so at the cost of further complicating attempts to represent interactions as originally hoped. But whether any of this is problematic, and if so to what degree, depends on one's hopes and expectations. \textit{What should count as a fully satisfactory theory? as a fully satisfactory theory of relativistic QFT?} 


In thus organizing the various responses to Haag's theorem, and in bringing the background commitments of the extra-Haagian outlooks into broad daylight, the framework itself does not adjudicate among the responses to Haag's theorem. Rather, it clarifies the various lines of genuine debate, both at the level of disagreement of technical matters of fact, and at the level of deeper methodological commitments. As such, these outlooks are directly implicated in \ref{sec:diagrams}.  Since there is no single ``wner' of QFT, let alone the edifice of science, the information that any inquirer of QFT will want\textemdash be they physicist, philosopher, or mathematician\textemdash will vary according to their own outlook. In our role as `general contractors', we indulge these owners by providing sample advice that either arranges the responses according to their outlook and even provisionally adopts an outlook.

A final note before we proceed. While the responses to Haag's theorem we review below proceed as if it presents a single problem demanding a singular response, this is not universal. In particular, \cite[Chs. 9--11]{ruetsche2011interpreting} notes several potential problems that unitary inequivalence phenomena, Haag's theorem among them, pose to a fundamental particle interpretation. Her response to these is pluralist insofar as she recognizes that ``\textit{different} interpretations will be indexed to---and adulterated by---different aspirations,'' where these aspirations ``adulterate because they arise \textit{along with} particular applications of the theory''  \cite[246]{ruetsche2011interpreting}. She argues that to demand a univocal interpretation, which would require a singular response to Haag's theorem, would ``stym[ie] attempts to harness the theory [...] in fruitful projects of explanation and expansion'' \cite[260]{ruetsche2011interpreting}. While this pluralist response is rich and worthy of discussion, we do not review it here because of its complexity and because it is not centered on Haag's theorem \textit{per se}.

\begin{table}[ptb]
    \centering
    \begin{tabularx}{\textwidth}{X*{5}{>{\centering\arraybackslash}X}}
        \toprule
        Name & \textbf{Assessment} & \textbf{Repair} & \textbf{Renovation or Maintenance} \\
        \midrule
        Earman \& D. Fraser & Interaction picture (IP) rests on a set of inconsistent  assumptions & Abandon the IP; adopt Haag-Ruelle scattering & Interpret and assess formal variations of QFT \\
        \midrule
        Theory Interpreters & In general no number operator for interacting QFT & Develop alternative ontology for QFT & Develop rigorous interacting QFT \\
        \midrule
         Duncan \& Miller & Pre-renormalization IP is inconsistent & Renormalization defuses Haag's theorem by breaking Poincar\'e invariance & Maintain current use of renormalization and regularization techniques\\
        \midrule
          Klaczynski & IP is inconsistent & Renormalization leads us to abandon the IP & Replace the IP's unitary intertwiner \\
        \midrule
        Maiezza \& Vasquez & The full perturbative series has renormalon divergences & Resummation methods (but these are ambiguous) & New resummation methods and physical insights are needed \\
        \midrule
        Kastner & Physics is non-local & Abandon notion of an independent field & Develop direct-action theories \\
        \midrule
        Seidewitz & Use of time as the evolution parameter is at fault &   Revise Wightman axioms to introduce new evolution parameter & Extend parameterized QFT to gauge theories and non-Abelian interactions\\
        \bottomrule
    \end{tabularx}
\caption{Framework application summary}
\label{tab:summary}
\end{table}

\subsection{Earman and D. Fraser}\label{subsec:EarmanFraser}
The philosophical literature on Haag's theorem is significantly influenced by  \cite{FraserDoreen2006Thesis} and \cite{EarmanFraser2006implications}. \cite{FraserDoreen2006Thesis} addresses the significance of Haag's theorem for interacting QFT, and \cite{EarmanFraser2006implications}  argue that there are multiple, significant  implications of Haag's theorem which are worthy of philosophical attention. However, the multiplicity of these implications all stem from the same root: Haag's theorem says that the interaction picture is predicated on an inconsistent set of assumptions. From this root, one stem concerns the implications for scattering theory, another holds implications for the use of non-Fock representations in describing interacting fields, and still another concerns the practitioner's choice among many unitarily inequivalent representations. 

\paragraph{Extra-Haagian outlooks.} 

D. Fraser's dissertation, \textit{Haag's Theorem and the Interpretation of Quantum Field Theories with Interactions}, set the stage for much of the philosophical literature on Haag's theorem over the last sixteen years (\cite{FraserDoreen2006Thesis}). We will therefore give a more detailed assessment of this source than subsequent works.  She writes, ``Haag’s theorem frames answers to the question `What is QFT?' '' (p. 171).  There are several contenders for the prize of being QFT. Fraser divides what we have called \textit{canonical QFT} into three subcategories: unrenormalized canonical QFT, renormalized canonical QFT, and canonical QFT with cutoffs. Contending against these for the prize of being QFT are several different formal variants of QFT: bottom up axiomatic work based on either the Wightman or Haag-Kastler axioms; constructive QFT; and algebraic QFT. Haag's theorem shows us that unrenormalized canonical QFT is predicated on an inconsistent set of initial assumptions, and is therefore disqualified. Neither renormalization nor the introduction of cut-offs are satisfactory in the face of this inconsistency: 

\begin{quotation}
 \noindent   We understand both renormalization and the introduction of cutoffs as expedients that permit the derivation of predictions for interacting systems, but do not address the root cause of the problem (namely, the inconsistency of the initial assumptions of the canonical framework). Renormalized canonical QFT and canonical QFT with cutoffs approximate (in some unspecified way!) a correct, completely finite theory; however, this theory is not string theory or some successor to QFT, but the correct formulation of QFT itself. \cite{FraserDoreen2006Thesis} (p. 174)
\end{quotation}

\noindent Thus, as we seek to answer the question ``What is QFT?'' D. Fraser argues that we eliminate contenders as follows: unrenormalized canonical QFT is disqualified for its inconsistency, demonstrated by Haag's theorem; renormalized canonical QFT is disqualified for being ``mathematically ill-defined'' (p. 173); whilst canonical QFT with cut-offs changes the goalposts. The goal was to integrate special relativity with NRQM, and that requires one to produce a theory on infinite, continuous space, ``not a lattice of finite spacial extent'' (p. 173). This leaves the formal variants of QFT as the only programs aimed at the ``correct formulation of QFT'' to which the other approaches are approximations. 

In subsequent work, D. Fraser is explicit about the role that correcting the inconsistency flagged by Haag's theorem plays in guiding the wise choice of variant of QFT to give philosophical attention: 

\begin{quotation}
 \noindent [C]onsistency is also a relevant criterion because quantum field theory is, by definition, the theory that integrates quantum theory and the special theory of relativity. Consistency is relevant to QFT for theoretical reasons\textemdash not for practical reasons (e.g., the derivation of predictions). As a result, it is necessary to either formulate a consistent theory or else show that this criterion cannot be satisfied (i.e., that there is no consistent theory with both quantum and special relativistic principles). \ldots The formal variant is the only variant that satisfies the criterion; its set of theoretical principles are both consistent and well motivated. \cite{fraser2009quantum} (p. 563)
\end{quotation}

\noindent Thus we have the key dialectical currents of the Fraser-Wallace debate fueling this line of response to Haag's theorem. Bedrock issues about what QFT even is, and what it takes to formulate it in a way that is apt for deep philosophical work on its conceptual-nomological machinery place Haag's theorem in the spotlight: \textit{this no-go theorem regiments answers to the central question at hand, namely, What is QFT?} In contrast, the typical theoretical particle physicist contents themselves with combining SR and QM to the best of their current ability, rather than demanding perfectly coherent and mathematically rigorous unification up front. 


Arguably, the best of their ability has improved over the decades such that, the average theoretical particle physicist could claim, we will inevitably \textit{find} such a complete and honest marriage \textit{if one is forthcoming}. In the meantime, we have successful, if incomplete, marriages of SR with QM with which we can keep experimental and theoretical particle physics going. While it is likewise a live option for D. Fraser (and presumably Earman) that an honest and complete marriage of QFT is not possible, this is a question of \textit{logic}, not practical or empirical feasibility.\footnote{\label{note:FraserAQFThope} Of course, D. Fraser's full corpus on the philosophy of QFT extends beyond the scope of issues in this present paper. In her dissertation she seemed optimistic that, with more time, the as-of-yet unfinished project of constructing a fully rigorous interacting QFT in four spacetime dimensions could be completed. Crucially, she allowed for the possibility that this will require revision to then-extant axiom systems (\cite{FraserDoreen2006Thesis} section 5). Subsequently, she put it to the proponents of renormalized and cutoff variants to wrestle with the question of whether such a completion is impossible: ``Neither the infinitely renormalized nor cutoff variant furnishes an argument that a consistent formulation of QFT is impossible; such an argument would require making the case that the axiomatic program cannot be completed.'' \cite{fraser2009quantum} (p. 563). See \cite{koberinski2023renormalization} for her most recent views on the epistemological status of effective field theory approaches in particle physics.}

We note briefly here that the physicists' attitude of settling for good-enough-to-be-getting-on-with variants of QFT clearly extends far beyond Haag's theorem. As just one other example, theoretical particle physicists are not bothered when they learn that their QFT path integrals cannot actually be integrals, given that there is no measure on the infinite-dimensional space of paths. Sorting out where path integrals properly live in the mathematical universe is a mathematician's job\textemdash not a physicist's. Meanwhile, that these path integrals are currently mathematically homeless in no way diminishes their calculational and conceptual fruitfulness for theoretical physics. Haag's theorem is one result among many indicating that the formal and canonical variants of QFT are not communicating well with each other: either, as Earman and D. Fraser see it, the foundations of QFT are in excellent shape, and they are waiting for more mathematical success in erecting interacting models of realistic matter in a realistic number of spacetime dimensions; or, as Wallace sees it,  QFT needs no foundational work since it does not really even need (mathematically) rigorous foundations. We, the authors, think that this communication breakdown between formal and informal variants of QFT is problematic. We will return to these concerns in section \ref{sec:diagrams}.

\paragraph{Assessment.}  \cite{EarmanFraser2006implications} give us the cold, hard logical truth:  Haag's theorem ``demonstrates that the interaction picture is predicated on an inconsistent set of assumptions. In response to this \textit{reductio} of the assumptions, at least one must be abandoned'' (p. 322). As discussed above, the goal for Earman and D. Fraser is to determine what QFT is, i.e., to completely and honestly unify SR and QM. A successful unification, if it is possible, will at least be able to model interacting systems. Thus, according to Earman and D. Fraser, Haag's theorem is a demonstration of the logical impossibility of modeling interactions using the interaction picture; all that the interaction picture is capable of modeling is, provably, free fields\textemdash a spectacular failure. ``On any reading, Haag’s theorem undermines the interaction picture and the attendant approach to scattering theory'' (p. 333).

\paragraph{Repair.} The repair work needed is  in the task of modeling interactions in QFT. This might proceed  either by abandoning the interaction picture altogether (e.g. with Haag-Ruelle scattering), or else by substantially revising it through giving up one or more of its assumptions. While assessing this second possible response, Earman and D. Fraser find the robustness of Haag's theorem, and its generalizations, to be significant: ``Subsequent no-go results do not show that field theorists do not have to worry about Haag’s theorem because some of its assumptions do not hold in all cases of interest; rather, what the subsequent results show [is] that even more assumptions have to be abandoned in order to obtain well-defined Hilbert space descriptions of interacting fields'' (318). Thus, Earman's and D. Fraser's preferred remedy for Haag's theorem is to abandon the interaction picture altogether in favor of alternative ways of modeling interactions in QFT. They discuss Haag-Ruelle scattering at length, presenting it as a readily available ``mathematically consistent alternative framework for scattering theory'' (p. 333). As a similar type of response, \cite{bain2000against}  argues that the innovations of the LSZ formalism provide the right solution to the heart of the conceptual problem posed by Haag's theorem\textemdash namely, ``the apparent incoherence of using the interaction picture in a situation in which its use dictates its non-existence'' (p. 383).\footnote{\cite{bain2000against} asserts that the LSZ formalism achieves this by replacing the strong convergence condition with weak convergence. \cite[note 23]{EarmanFraser2006implications} disagree, arguing that the information contained in the in and out states of the theory is already contained in the behavior of the interpolating Heisenberg field of Haag-Ruelle. Rather, the weak convergence condition of LSZ takes the further step of explicitly linking the S-matrix and vacuum expectation values of the free field. See \cite[275-289]{duncan2012conceptual} for more discussion of the relationship between LSZ and Haag-Ruelle.}

Note that this accords with Earman and  D. Fraser's stated goal of unifying SR and QM. Regularization and renormalization techniques are not options available to them as they change the game or execute the unification dishonestly, respectively. Thus, the only option is to give up the interaction picture: Haag's theorem does not pose a problem for QFT per se, but it ``does pose problems for some of the techniques used in textbook physics for extracting physical prediction from the theory'' \cite{EarmanFraser2006implications}(p. 306). QFT, therefore, is not to be identified with textbook physics. The textbook physics is one attempt at doing QFT, and Haag's theorem exposes serious problems in that particular attempt. The immediate technical moral of Haag's theorem is that (better attempts at) QFT must embrace the use of unitarily inequivalent representations of the CCR.\footnote{``[A] single, universal Hilbert space representation does not suffice for describing both free and interacting fields; instead, unitarily inequivalent representation of the CCR must be employed'' (p. 333).} The philosophical insight from this technical result is that it designates the role of unitarily inequivalent representations as a distinctive feature of QFT in contrast to QM.\footnote{``Haag's theorem was instrumental in convincing physicists that inequivalent representation of the CCR are not mere mathematical playthings but are essential in the description of quantum fields'' (p. 319). }

In addition, Earman and D. Fraser recognize that an explanation is needed for the success of the interaction picture and perturbation theory in the face of Haag's theorem. They  ``suspect that the full explanation will have a number of different pieces'' and offer their analysis of Haag-Ruelle scattering as one such piece (p. 322). The Duncan and Miller response discussed  below (section \ref{subsec:duncanandmiller}) should be thought of as adding one additional such piece, not refuting Earman and D. Fraser.

\paragraph{Renovation.} As a consequence of this logical diagnosis and logical remedy, Earman and D. Fraser advocate for future philosophical work and resources to be deployed in philosophical projects about formal variants of QFT and their attendant non-interaction picture approaches to interactions. There are at least two such broad philosophical projects that benefit from the ways in which practicing mathematicians have ``digested the lesson of Haag's theorem'' \cite{EarmanFraser2006implications} (p. 334). First there is the project of interpreting the mathematical structures used in algebraic QFT and in constructive QFT. Second, and relatedly, there is the subtle matter of assessing the implications of theorems of axiomatic approaches to QFT for the philosophy of those areas of fundamental physics that make use of QFT.

\subsection{Theory Interpreters\textemdash The Particle Problem}\label{subsec:particleproblem}

There is a segment of the philosophical literature addressing Haag's theorem that is primarily motivated by the question, as put by Ruetsche, ``Is particle physics particle physics?'' \cite[190]{ruetsche2011interpreting}. The issue is whether or not the quantum field theories employed in particle physics admit of particle interpretations. The lines of argument for and against particle interpretations are numerous, and we direct interested readers to \cite{ruetsche2011interpreting} ch. 9, and \cite{fraser2021particles} as entry points to these debates. Here we focus on the hurdle that Haag's theorem poses for particle interpretations.

\cite{Fraser2008-pg} advances the argument from Haag's theorem to the elimination of particles from our ontology. (\cite[20]{halvorson2002no} briefly sketch the initial segment of this argument.) The heart of the argument runs, roughly, as follows. At a minimum, a particle interpretation requires that the theory admits of a formulation in which a number operator exists\textemdash you cannot very well say that the theory says it is possible that matter is fundamentally particles if the theory also says that it is not possible (in theoretical principle) for those particles to be counted. The Fock space representation of a free field yields a well-defined number operator as well as the correct relativistic energies suitable for a particle interpretation. Might we then, by means of the IP, use the Fock space of a free field for representing and counting particle states in an interacting system as well? Haag's theorem blocks this path, since it entails that representations of the ETCCRs for free and interacting systems are unitarily inequivalent. For this and other reasons (see Assessment below), Fraser, Halvorson, and Clifton conclude that a particle interpretation for QFT is a no-go, and so there are no particles. Others  (\cite{wallace2011taking} and \cite{bain2011quantum}) disagree and seek to rehabilitate an emergent particle interpretation despite Haag's theorem.



\paragraph{Extra-Haagian outlooks.} 


This sub-literature is unified by a decidedly \textit{interpretive} task: the authors aim (in this context) to answer questions about what ontological commitments are supported by QFT. Thus, this literature is not ultimately motivated by considerations of theory consistency, application, development and innovation, or conceptual-nomological machinery (``how does this theory work?'') that guide the work of groups surveyed in other subsections. Although such considerations sometime arise in this literature (as, indeed, matters of theory interpretation sometimes arises as a secondary consideration in other segments of the Haag's theorem literature), they do so as steps within larger arguments ultimately aimed at interpretation. (And of course, many authors take up distinctive projects across this interpretation/not-interpretation divide.)

There is a further extra-Haagian outlook operative under the surface of this literature. This outlook concerns the basic nature of the question at stake. For some, the question to be answered is circumscribed as a highly-constrained puzzle: given QFT in a mathematically rigorous form, is it possible to construct a number operator for interacting QFT? Others, in contrast, aim to assess the viability of a particle interpretation understood much more loosely and flexibly from the outset. \cite[857]{fraser2009quantum} explains this difference (situating herself in the former camp) as coming down to solving a problem within a research program versus setting out on a different program altogether: those who, like Wallace and Bain (and in a different way  \cite{feintzeig2021localizable}), aim at rehabilitating (our talk of) particles under a revised, emergent, or approximate notion of particle offer ``a program, not a solution'' to the original puzzle. Ruetsche agrees with Fraser's analysis of the way in which Haag's theorem precludes a particle interpretation of interacting QFT writ large. However, Ruetsche argues for a multifaceted, less ``pristine'' approach to theory interpretation. For Ruetsche, the argument from Haag's against particles applies in certain context, but not in others, such that ``[s]ometimes particle physics is, adulteratedly, particle physics, and that's a good thing'' \cite[260]{ruetsche2011interpreting}.

Thus, there are extra-Haagian disagreements over the method of theory interpretation. There is also an extra-Haagian disagreement over the standards of theory formulation required before the interpretive task can begin. For \cite{halvorson2002no} and \cite{Fraser2008-pg}, the theory to be interpreted \textit{is} the formal, mathematically well-posed version of the scientific achievement under study (\cite{fraser2009quantum}, \cite{fraser2011take}). No-go theorems proved within these formal theoretical structures, then, powerfully and persuasively constrain interpretive options  \footnote{See \cite{FGMhowto} for further discussion on alternative stances toward no-go theorems in general, and Haag's theorem in particular.} \cite{Fraser2008-pg} is explicit on this point. In describing the role of Haag's theorem for blocking one attempt at developing a particle interpretation of interacting QFT, she writes: 

\begin{quotation}
  \noindent  Haag’s theorem pinpoints the source of the problem with the strategy of obtaining a quanta interpretation for an interacting system from the Fock representation for some free system. The consensus among axiomatic quantum field theorists is that Haag’s theorem entails that a Fock representation for a free field cannot be used to represent an interacting field. (p. 847)
\end{quotation}
Moreover, according to Wightman, it is a ``\textit{necessary} condition'' that any physically sensible interacting field theory use a representation of the ETCCRs that is unitarily inequivalent to a Fock representation.
Thus, Haag's theorem poses a problem for a particle interpretation provided that one is committed to mathematical rigor in theory formulation and interpretation. 
\footnote{See  \cite{halvorson2007} p. 731-732 for a concise statement of this attitude towards interpretation and mathematical precision in theory formulation.}

In contrast, \cite{bain2000against} and \cite{wallace2001emergence} take a  more flexible and opportunistic approach to the question of particle interpretations for QFT. To them, Wightman's ``necessary condition'' looks more like an interesting suggestion or a useful starting point. From such a starting point, they turn their attention toward opportunities within the messy details of how interacting QFT gets applied to develop a serviceable notion of particles without meeting the stringent requirements of finding a representation unitarily equivalent to a Fock space representation. The usual response along these lines appeals to a scale-variant or emergent particle interpretation (\cite[124]{wallace2011taking}, \cite{wallace2001emergence})\footnote{\cite{wallace2001emergence} aims to give a concrete treatment of the emergence of particle structure from quantum fields. He aims to explain how the particle concept comes to be useful for interpreting a theory that seems \textit{prima facie} to be about very un-particle-like entities known as fields. This is not a solution to the highly-constrained puzzle (constructing a number operator for interacting QFT) that \cite{Fraser2008-pg} and \cite{halvorson2002no} aim to address.}, or else to an asymptotic notion of particle states within the interacting theory (\cite{bain2000against}).  
While this difference in outlook is often implicit, D. Fraser (in response to \cite{wallace2001emergence}) is transparently skeptical of even the cogency of an ontology with non-fundamental entities or which is not built on an exact similarity relation with the mathematical model at hand \cite[857--8]{Fraser2008-pg}.

\paragraph{Assessment.} For \cite{Fraser2008-pg} and \cite{halvorson2002no}, Haag's theorem is a diagnostic no-go result concerning interpretation: QFT cannot be interpreted as a theory of particles. For instance \cite[23]{halvorson2002no} appeal to Haag's theorem, among other no-go results, when they conclude that the theory of quantum fields ``does not permit an ontology of localizable particles; and so, strictly speaking, our talk about localizable particles is a fiction'', useful though this fiction may be. As \cite[20]{halvorson2002no} briefly describe, Haag's theorem precludes a particle ontology as follows. One first requires that there is a representation of the Weyl relations having a global occupation number operator. These are usually furnished with the interpretation of counting the number of particles in a state (i.e., it is the particle representation discussed in \ref{subsec:history}).\footnote{See \cite[845-846]{Fraser2008-pg} for an exposition as to why these operators at apt for counting particles, rather than counting energy levels or other discrete quantum quantities.} Since number operators only exist for representations unitarily equivalent to a free-field Fock representation (theorem 3.3 of \cite{chaiken1968number}), and since Haag's theorem implies that the latter cannot represent an interacting field, interacting field theory does not admit of a particle interpretation.

\cite[847-849]{Fraser2008-pg} demonstrates the force of Haag's theorem against particle interpretation in considerably more detail. She describes three distinct attempts to generalize a quanta interpretation of free fields to interacting fields, arguing that each of these fail. The first of these attempts is foiled by Haag's theorem; the others fail for different reasons.\footnote{But they, like the approach in \cite{bain2013emergence} to the particle problem, count the evasion of Haag's theorem among their merits.} The first attempt tries to rather straightforwardly extend the notion of particle available to free field theories to interacting ones, just as \cite{halvorson2002no} discuss. In the free case, we obtain a well defined number operator, whose eigenvectors may be interpreted as counting definite numbers of particles so long as we are in a Fock representation. But Haag's theorem shows us that we cannot use a Fock representation of free fields to represent interacting ones.\footnote{See also \cite{heathcote1989theory} pp. 91-97 for a discussion of the consequences of Haag's theorem for Fock space representations and axiomatic approaches to modeling interactions, and see \cite{EarmanFraser2006implications} section 7 for their critical assessment thereof.} And of course, when we are looking for a particle interpretation of QFT, we care significantly more about the interacting case than we do the free case. Given this argument \cite{halvorson2002no} and \cite{Fraser2008-pg} conclude that the fundamental ontology of the actual world does not include any particles. In contrast, \cite{bain2013emergence}, \cite{wallace2001emergence}, and \cite{feintzeig2021localizable} aim (in their own distinctive ways) to recover a non-fundamental, emergent notion of particles.

Thus, for these authors, Haag's theorem sits alongside results like the Reeh-Schlieder theorem \cite{reeh1961bemerkungen}, the Unruh effect (\cite{Fulling1973PhysRevD.7.2850}, \cite{Davies1975}, \cite{Unruh1976PhysRevD.14.870} and discussed more recently in \cite{crispino2008unruh} \cite{earman2011unruh}),  or the Hegerfeldt \cite{hegerfeldt1998causality} \cite{hegerfeldt1998instantaneous} and Malament \cite{malament1996defense} theorems (and extensions thereof \cite{halvorson2002no}) as evidence against a particle interpretation for RQFT. Indeed, canon seems to dictate that these results be discussed together when addressing the particle problem (e.g., \cite[380]{bain2000against} \cite[20]{halvorson2002no} \cite[842]{Fraser2008-pg} \cite[Chs. 9--11]{ruetsche2011interpreting} \cite[\S 2.4]{wallace2001emergence}.)



\paragraph{Repair.} Interpreters addressing the highly-constrained puzzle of constructing a number operator for interacting QFT bite the bullet and accept that there are no fundamental particles. Doing so entails providing an alternative fundamental ontology and alternative explanations of erstwhile ``particle'' physics phenomena. Halvorson and Clifton conclude that ``relativistic quantum field theory does permit \textit{talk} about particles\textemdash albeit, if we understand this talk as really being about the properties of, and interactions among, quantized fields'' \cite[24]{halvorson2002no}. A large part of the repair offered in \cite{halvorson2002no} is an account of how our talk about ``particles'' is still useful, despite its fictitious status. Similarly, albeit with less patience for the talk of particles, D. Fraser concludes that special relativity and the non-linearity of the field equation for an interacting system\textemdash key ingredients for Haag's theorem\textemdash conspire against a quanta interpretation, lending some weak support to a field interpretation instead (``At least on the surface QFT is a theory of fields'' and ``there is no quanta interpretation and there are no quanta'' \cite[857-8]{Fraser2008-pg}.\footnote{However, it should be noted that the most obvious way of constructing a field ontology appears to run into the same problems as the particle ontology \cite{baker-fields}. Baker's argument relies on the restriction that the field wave functionals are square-integrable. If we relax this restriction, then the space of wave functionals will be larger than the space of particle wave functions (see \cite{Jackiw1990} and \cite{Chip2022}, but see \cite{Wallace2006_naivite} for a defense of the square-integrability restriction).}

If, however, we set aside the highly-constrained puzzle within the program of rigorous formulations of QFT, we can pursue alternative programs of particle interpretations rooted in revised formal notions of `particle.'  These alternative programs can be pursued according to various standards of theory interpretation, but in any case one must revise the relevant definitions at play in the no-go arguments for particles. One such option revises the basic approach to scattering theory (Haag-Ruelle) \cite{reedsimon_ii}. In this approach, one jettisons the interaction picture, and instead finds surrogates for interacting states among free states. Thus, the challenges of accessing a Fock space representation for the interacting theory disappear.  Another option, the LSZ formalism, revises the relations between the interacting field and the asymptotic free fields. As shown in \cite{bain2000against}, the LSZ formalism defuses the argument that the interacting field must be free by replacing strong convergence with weak convergence. Thus, Bain argues that the LSZ formalism solves the conceptual problem posed by Haag's theorem. Bain formulates asymptotic raising and lowering operators by taking appropriate limits of the \textit{interacting} raising and lowering operators. These asymptotic operators are then used to define asymptotic particle states. The asymptotic particles states are \textit{not} free field states, and so in principle there are contributions from the multi-particle continuum. Bain (\S 4.2) shows that these contributions are negligible on the relevant time scales, and thus, for all practical purposes, the asymptotic particle states replicate the behavior of free particle states. On the LSZ approach, then, Haag's theorem is circumvented through the move to weak convergence; and the expected phenomenological difference between interacting states and free-like asymptotic states is recovered. The result is a rehabilitation of a particle interpretation through the redefinition of asymptotic states. 

Or we might instead consider an emergent notion of particles in QFT. Instead of proposing repairs to either our ontology or the assumptions of Haag's theorem, for \cite{wallace2001emergence} it suffices to show that a concept of particle can be recovered from a field-theoretic description. Given Wallace's background presumption that approximate concepts are legitimate in physics, it does not matter (as much) that the particle concept he deploys is vaguely defined and applies only to bosonic, massive QFTs.


Finally, for Ruetsche, Haag's theorem shows that, ``within the confines of the interaction picture'' there is no  particle interpretation that can account for the entire micro history of a scattering experiment. But Ruetsche draws attention to those confines and their surroundings, assessing the viability of a particle interpretation in non-interaction picture contexts for QFT scattering theory ($\S$11.1.2) and in cosmological particle creation ($\S$11.2). The areas of physics germane to the question of particle ontologies are far and wide, and there is no reason to think that the Haagian roadblock to a particle interpretation should have the final say. Indeed, Ruetche concludes, ``[s]ometimes particle physics is, adulteratedly, particle physics, and that's a good thing'' \cite[260]{ruetsche2011interpreting}.

\paragraph{Renovation.} Here theory interpreters especially diverge from the rest of our authors. For both those attendant to the highly-constrained puzzle \textemdash given QFT in a mathematically rigorous form, is it possible to construct a number operator for interacting QFT?\textemdash and those who are not, the interpretive projects seek to establish an ontology conditional on an extant form of QFT. This interpretive task is a step removed from questions about how to develop or innovate with our theory. As such, the renovations that theory interpreters propose to improve interpretation will be mediated by how they think the theory should be developed. Given that this latter question falls outside the domain of theory interpretation \textit{per se}, theory interpreters often say comparatively little about what renovations are necessary in response to Haag's theorem.

Nevertheless, we can say a little bit about what renovations theory interpreters might call for. For those trying to solve the puzzle, we need a rigorous formulation of QFT or its successor which we can subject to interpretation. To answer the puzzle, this formulation must be general enough that we can draw interpretive conclusions about all of the phenomena we consider relevant, not just a subset. These calls would seem to make puzzle solvers inclined to axiomatic approaches to theory renovation and, in general, to development of formal variants of QFT.\footnote{Recall D. Fraser's assessment of the search for a mathematically rigorous model of QFT in note \ref{note:FraserAQFThope}. One noteworthy recent development on what kinds of field theories can be rigorously constructed in four spacetime dimensions is \cite{Aizenman-Duminil-Copin-2021}.} On the other hand, those less dedicated to the original puzzle seem to call for renovations of the basic conceptual apparatus, especially the notion of particle. \cite{bain2011quantum} argues that the intuitions of localizability and countability constitutive of the particle notion at stake in these debates are themselves non-relativistic, and urges us to develop one more appropriate for relativistic theories. Similarly, interpretive projects aimed at explicating an emergent notion of particles (\cite{wallace2001emergence}, \cite{feintzeig2021localizable}) advocate long term renovation work in terms of conceptual criteria for particle interpretations.

\subsection{Duncan and Miller}\label{subsec:duncanandmiller}

Tony Duncan \cite{duncan2012conceptual} and Michael Miller \cite{miller2018haag} argue that at least one of the assumptions needed to prove Haag's theorem is violated at some point in the actual process of calculating scattering theory results in perturbative QFT. Thus, in Duncan's words, we can ``stop worrying'' about Haag's theorem. These calculational processes include the methods of regularization and renormalization. \footnote{Regularization is a mathematical procedure used to temporarily control the ultraviolet (UV) or infrared (IR) divergences that arise in QFT calculations. After regularization, it becomes possible to extract finite results from the perturbative series, order-by-order. However any regularization technique will introduce new, arbitrary parameters (regularization scales). Renormalization is a subsequent procedure in which the arbitrary dependence on these regularization scales is absorbed by redefining the bare parameters of the theory. For a fuller explanation, see \cite{collins1984renormalization}. }

A similar view on Haag's theorem is given in \cite{FraserJAMES2017RealProblem}. Here, James Fraser is primarily concerned with the much broader question of the status of perturbative QFT and diagnosing the ``real problem'' with this area of physics. He argues that perturbative QFT is ``a method for producing approximations without addressing the project of constructing interacting QFT models'' (4). \footnote{Miller also addresses the question of what happens after we remove UV and IR cutoffs and thus restore exact Poincare covariance, which J. Fraser does not directly discuss.} As a consequence of this view, the threat of inconsistency posed by Haag's theorem is defused for much the same reasons as given by Duncan and Miller. J. Fraser concludes section 4, ``The perturbative method simply does not assert the set of claims shown to be inconsistent by Haag's theorem'' (18). We set aside \cite{FraserJAMES2017RealProblem} for the remainder of this section since Haag's theorem is not the primary target of that article; but readers interested in the broader question of how to assess perturbative QFT are encouraged to look there for an important contribution on that topic.

\paragraph{Extra-Haagian outlooks.} For Duncan and Miller, it is a given that the interaction picture has been consistently applied in calculating various specific theoretical predictions. For Duncan, QFT as it is used in particle physics---including its use of the interaction picture---is ``the most powerful, beautiful, and effective theoretical edifice ever constructed in the physical sciences'' \cite[iv]{duncan2012conceptual}. His goal, therefore, cannot coincide exactly with Earman and D. Fraser's of strictly unifying SR with QM. Rather, Duncan's broader goal for the book is to provide a ``deep and satisfying comprehension'' of QFT by addressing the important conceptual issues for which the traditional, ``utilitarian'' texts fail to provide careful explanations (pp. iii--iv).

Likewise, Miller is focused on understanding the QFT noted for its empirical successes \cite[802]{miller2018haag}, hence not the QFT of Earman and D. Fraser. Ultimately, Miller's aim is to address ``a general tension that exists in much of the literature engaged in the philosophical appraisal of the foundations of quantum field theory'' (p. 803). This tension is essentially that between Earman and D. Fraser and the theory interpreters' approach to QFT on the one hand and the empirically-tractable QFT on the other: while the former is mathematically rigorous and hence (relatively) straightforward to interpret using standard philosophical tools, it has yet to produce a realistic model, so it is unclear how it could inform claims about the actual world; conversely, while the latter has generated wildly successful empirical predictions, it has done so through changes to the mathematical formalism whose interpretive significance is far from obvious (p. 803).

For both authors, then, the preferred approach seems to be to bring the philosopher's penchant for logical and conceptual clarity to bear on  QFT \textit{as it is actually used}. Both are sure (up to fallibility) that particle physicists succeed at applying QFT for predictions for scattering processes. From this standpoint, Duncan and Miller seek to account for \textit{how} the use of the interaction picture in QFT in practice overcomes the challenge of Haag's theorem, rather than questioning \textit{whether} it can.


\paragraph{Assessment.} 

Both Duncan and Miller recognize the logical nature of the problem posed by Haag's theorem\textemdash i.e., as the conflict between the interaction picture and free-field representation assumptions\textemdash and they take consistency to be a requirement for a viable theory. Indeed, they recognize that a response to Haag's theorem must come from negating at least one of its assumptions. However, unlike Earman and D. Fraser, their goal is to ``explain why theoretical predictions for realistic experimental observables give empirically adequate results'' \cite[803]{miller2018haag}. These predictions \textit{in fact} apply (a version of) the interaction picture, so they cannot simply jettison the framework of the interaction picture altogether. Rather, they aim to identify where within the practical calculational techniques such violations of the assumptions must already take place for some other reasons unmotivated by Haag's theorem.


\cite[pages 359-370]{duncan2012conceptual} and \cite{miller2018haag} each argue that we circumvent Haag's theorem in the messy calculational details of how the interaction picture is used in practice.  In fact, they identify several areas which might demand explanation in the light of Haag's theorem.

First, is the question of why perturbative calculations, built upon the interaction picture, are able to yield empirically adequate results at each finite order. Duncan and Miller argue this takes place in a two-step process.

In the first step, the regularization procedure evades Haag's theorem.  Both Duncan and Miller focus on the case of cutoff regularization, in which short and long distance cutoffs (or, equivalently, high and low momentum cutoffs) are introduced, reducing the theory to a finite number of degrees of freedom. Such a procedure clearly breaks the Poincar\'e invariance (and generally the unitarity) of the S-matrix, and so the regularized theory evades the assumptions of Haag's theorem. 

In the second step, renormalization, the theory's parameters are redefined to remove the dependence on the cutoff scales that were introduced during regularization. As such the Poincar\'e invariance and unitarity of the theory are restored. Each individual term in the renormalized perturbative series is well-defined, and the series can be used up to finite order to make calculations that are compared to experiment. Thus according to \cite[370]{duncan2012conceptual}, ``the proper response to Haag’s theorem is simply a frank admission that the same regularizations needed to make proper mathematical sense of the dynamics of an interacting field theory at each stage of a perturbative calculation will do double duty in restoring the applicability of the interaction picture at intermediate stages of the calculation.''


Note that the violation of Poincar\'e invariance cannot be the explanation in all cases. Many, widely used methods of regularization (such as dimensional regularization) do not break Poincar\'e invariance. This raises the question as to how this story should apply in such cases. In response, \cite[pages 814-815]{miller2018haag} suggests that the basic structure of the argument still stands. Any such regularization technique must provide \emph{some} means for controlling infrared divergences. In so doing, one expects that one or more of the assumptions of Haag's theorem \emph{must} have been violated through the regularization process.

However this raises a second question: after renormalization, the full symmetries of the theory are restored, so why do we continue to get empirically adequate results? After all, restoring such symmetries should once again render the entire formalism ill-defined. Here, \cite[p. 815]{miller2018haag} argues, ``The best available explanation of this fact is that the observables that get compared to experiment are insensitive to the removal of the infrared cut-off.'' \footnote{Miller points to the KLN theorem of \cite{KLN1, KLN2}. The KLN theorem ensures that IR divergences, which appear in individual terms of the perturbative series, cancel out when summing over all terms. The cancellation happens when considering the entire series of terms, not in each individual term. So, while each term might have an IR divergence, the full result, corresponding to a physically measurable quantity, does not.} As \cite[1134]{miller2021infrared} explains, this involves expressing physical quantities in terms of detector resolution, which conditions ``on which kinds of soft radiation can be undetected in the final state.'' This conditioning is appropriate insofar as it merely articulates ``the precise nature of the question we are asking about the outgoing state by executing the particular measuring process that we chose to execute.'' While exactly which assumption is given up varies, the standard toolkits used by physicists to generate practical calculations nevertheless allow them to avoid any possible traps laid by Haag's theorem. 

Even with this in place, Duncan and Miller point to a further theoretical problem with the perturbative series. Although the individual terms in the properly renormalized perturbative series are finite, the series as a whole still diverges \cite[187]{miller2018haag} \cite[374-411]{duncan2012conceptual} \footnote{Arguably, there is still significant explanatory work to do regarding those cases in which the perturbative series is not Borel-summable.}. However, Duncan does not attribute this problem to the consequences of Haag's theorem and Miller goes further arguing the problem is ``not related to Haag's theorem'' \cite[817]{miller2018haag}. Thus for both Duncan and Miller, whatever theoretical problems may remain with the perturbative series, the issue of how it can evade Haag's theorem in particular has been solved.



\paragraph{Repair.} The interaction picture in practice makes use of renormalization and regularization techniques that, in pulling ``double duty'', both produce finite perturbative results and defuse the assumptions of Haag's theorem. The renormalization and regularization repairs also evade Haag's theorem as a positive side effect. No additional repair 
is needed.

\paragraph{Maintenance.} Maintenance is straightforward: maintain current practices of proper use of renormalization and regularization techniques. For the conceptually and logically curious, following in Duncan and Miller's footsteps, this maintenance will likely also involve identifying precisely which assumptions have been violated, and where. While it remains to be shown that the perturbative power series converge and thus correspond to exact non-perturbative objects post-regularization and -renormalization, \cite[page 815]{miller2018haag} contends that this problem, while more serious, is ultimately unrelated to Haag’s theorem (compare \cite[1135]{miller2021infrared}).

\subsection{Klaczynski}

\paragraph{Extra-Haagian outlooks.} 
 Like the authors above, Klaczynski recognizes a bifurcation in QFT research. On the one hand there is canonical QFT. Canonical QFT has been spectacularly successful not only for making precise predictions but, as Klaczynski emphasizes, ``\textit{predict[ing] the existence of hitherto unknown particles} \cite[2]{klaczynski2016haags}. Nevertheless, canonical QFT is mathematically ill-defined: ``canonical QFT presents itself as a stupendous and intricate jigsaw puzzle. While some massive chunks are for themselves coherent, we shall see that some connecting pieces are still only tenuously locked, though simply taken for granted by many practising physicists, both of phenomenological and of theoretical creed'' \cite[2]{klaczynski2016haags}. One major contributor to this ill-definedness is owed to the use of renormalization. 

On the other hand is constructive QFT, which use operator theory and stochastic analysis to attempt to construct models of quantum fields in a mathematically well-defined manner. A number of important results have been obtained by this research program, including many triviality results that can be seen as calling into question basic features expected of any rigorous QFT. Nevertheless, the construction of a renormalizable theory in 4 dimensions---i.e., realistic---has neither been achieved nor seems achievable using current methods \cite[3]{klaczynski2016haags}. That is, no rigorous and realistic model exists. Klaczynski's aim is to reconcile canonical and constructive QFT by elucidating the coherence brought about by renormalization.

\paragraph{Assessment.}
Given the above, Klaczynski approaches Haag's theorem intent on understanding what it says about canonical QFT. At first blush, the result appears to be negative. For instance, the Gell-Mann and Low formula, relating the ground states of the interacting and non-interacting fields, is built on the assumption that the time-evolution operator in the interaction picture, which relates the two fields, is unitary; this is exactly what Haag's theorem rules out (recall note \ref{note:GellMannLow} on the significance and widespread use of the Gell-Mann and Low formula.). However, on closer inspection, the contradiction is resolved, if nevertheless unfavorably: like J. Fraser and Duncan, Klaczynski blames Haag's theorem for the divergence of the perturbative expansion of the Gell-Mann and Low formula. This leads Klaczynski to conclude that the interaction picture is ill-defined and trivial \cite[59]{klaczynski2016haags}.

While he points to similar symptoms as Duncan and Miller, Klaczynski comes to a different conclusion. In his final assessment, the interaction picture itself, relying as it does on a unitary intertwiner, is flawed---even renormalization does not save the interaction picture. While standard regularization methods may break Poincar{\'e} invariance, this is physically unacceptable; moreover, Poincar{\'e} invariance broken by regularization is restored when we take the adiabatic limit. This means that Haag's theorem again applies, albeit now to the renormalized theory \cite[62--3]{klaczynski2016haags}.

\paragraph{Repair.} Like Duncan and Miller, Klaczynski thinks that renormalization procedure still repairs the problem, but gives a different reason for why it works. When we renormalize our theories,  we replace the bare quantities with their renormalised counterparts. This process of renormalization does not merely cancel the divergences, but also fundamentally changes the theory, by bringing about a coupling-dependent mass shift. As a result, the renormalized free and interacting fields now have different masses. Two quantum fields of different mass are overwhelmingly likely to be unitarily inequivalent at any given time. \footnote{ \label{ft:ReedSimon} The mathematical ill-definedness of the renormalized terms makes it hard to prove that renormalized free and interacting fields will be unitarily inequivalent in general. However, Klaczynski uses a theorem X.46 of \cite{reedsimon_ii} to argue that it is ``plausible beyond doubt'' \cite[page 68]{klaczynski2016haags}. Theorem X.46 says that if there are two free, scalar, neutral fields with a unitary intertwiner at some time $t_0$, with \emph{different} masses, then these fields must be unitarily inequivalent. Klaczynski uses the theorem to argue we should not expect a unitary intertwiner at any time between the renormalized interacting and free fields. Recall that Haag's theorem requires precisely such a unitary equivalence between the free and interacting fields at some time (\ref{eqn:unitaryrelationbetween2fields_2}). }, and so the conditions for Haag's theorem do not apply in our renormalized theories.

Nevertheless, the un-renormalized and renormalized theories \textit{are} still related---just not by a unitary transformation, as physicists still believe. Thus, we are working with something \textit{like} the interaction picture insofar as the two theories are still related by an intertwiner of sorts, but \textit{unlike} the interaction picture insofar as this intertwiner is manifestly non-unitary. So whilst renormalization is the correct treatment, physicists have not fully grasped how it allows us to evade Haag's theorem. By using renormalization to fix the divergences,  physicists have ``muddled our way through to successfully applying perturbation theory'' [Klaczynski, personal communication, Sept. 8, 2021]. 


\paragraph{Renovation.} Given that the interaction picture presentation of canonical QFT still dominates, renovation is necessary. This has two parts. First, physicists have misunderstood, or at least misdescribed, the tools they are using when performing scattering calculations with renormalized fields. While Haag's theorem says this tool cannot be the interaction picture, that does not mean renormalization is unconstrained or incoherent. Rather, renormalization ``follows rules which have a neat underlying algebraic structure [the Hopf algebra] and are not those of a random whack-a-mole game'' \cite[page 4]{klaczynski2016haags}.

Second, this program of identifying the structure of renormalization must continue. First and foremost, there are some mathematical lacunae in this process, as well as procedures that are not defined wholly rigorously; these should be addressed. Klaczynski believes that renormalized quantum field theory ``provides us with peepholes through which we are allowed to glimpse at least some parts of that `true' theory' '' \cite[page 4]{klaczynski2016haags}. \footnote{This view can be fruitfully compared to that of \cite{FraserJAMES2017RealProblem}. Both agree that the perturbative renormalization offers a consistent calculational procedure that can both generate empirical predictions and offers some insights about underlying ``true'' physics. However, J. Fraser argues for a selectively realist attitude towards some parts of the effective theory. Klaczynski seems to argue only that this theory might provide some insights about a ``true'' physical theory, one which gives a true representation of the physical system.} 





\subsection{Maiezza and Vasquez}

\paragraph{Extra-Haagian outlooks.}

Like Klaczynski, Maiezza and Vasquez are interested in determining the precise mathematical structure of canonical QFT. However, they seek a ``consistent and generic (non-perturbative) formulation of QFT'' \cite[2]{maiezza2020haag}, in contrast to the perturbative formulation of canonical QFT. Given the centrality of the interaction picture to canonical QFT, the natural starting point for finding such a consistent and generic formulation of QFT would be to use the interaction picture, starting with free fields. But, Maiezza and Vasquez ask, is it possible to build such a non-perturbative formulation from the interaction picture and perturbative renormalization? Or, rather, is something new needed?

\paragraph{Assessment.}
In Maiezza and Vasquez's assessment, Haag's theorem reveals a central problem when we try to improve upon the standard methods of perturbatively renormalized canonical QFT. Haag's theorem is generally associated with infrared divergences of individual terms in the perturbative series (see \cite[p. 319]{EarmanFraser2006implications}, \cite[pp 63-66]{FraserDoreen2006Thesis}, \cite[p. 34]{Sbisa:2020whw} ). However, Maiezza and Vasquez trace the \emph{overall divergence} of the perturbative series back to Haag's theorem. After all, ``Haag’s theorem is non-perturbative and thus, for a finite value of the coupling,  perturbative renormalization cannot be its solution'' \cite[page 7]{maiezza2020haag}.

We see this problem clearly when we consider the entire perturbative expansion series. The problem stems from vacuum polarization: because of vacuum polarization, the full power series expansion diverges\footnote{This was first noted in \cite{Dyson_divergent} concerning QED.}. Maiezza and Vasquez demonstrate that renormalon singularities arise in the total perturbative series, and trace these back to Haag's theorem. These renormalon singularities ``are the concrete manifestation'' of the impossibility of generating an unambiguous finite result in such cases \cite[page 10]{maiezza2020haag}. 
\footnote{
Renormalons are singularities that arise as a function of the complex transform parameter when a formally divergent series is summed using Borel summation, see \cite{Beneke_renormalons} for a review. Infrared renormalons are generally associated with long-distance or low energy non-perturbative effects, emerging from the accumulation of low-momentum virtual particles in loop corrections, whilst ultraviolet renormalons are generally associated with artefacts introduced by the renormalization procedure itself (see also \cite{Parisi_Renormalons_1979}). 
}
For non-asymptotically free models, Maiezza and Vasquez associate the UV renormalon divergences with Haag's theorem. They argue that the same would hold for IR renormalons in asymptotically free models.\footnote{Roughly speaking, in an asymptotically free quantum field theory, the interaction strength becomes weaker as the distance scale decreases or the energy scale increases. See \cite{Politzer_AF_1973, Gross_Wilczek_AF_1973, Politzer_AF_1974, Wilczek_2005} for further details.} Thus, either kind of renormalon divergence can be traced back to Haag's theorem in some cases. \footnote{Maiezza and Vasquez' make no de facto distinction between UV and IR renormalons (c.f. \cite{Parisi_Renormalons_1979}), instead showing how renormalons emerge from the renormalization group equations within the resurgence framework, without reference to specific Feynman diagrams (see also \cite{Maiezza_resurgence, Maiezza_nonwilsonian, Dunne:nonperturb}). As such, either UV or IR renormalons may be interpreted as a symptoms of Haag's theorem, depending on the theory. For instance, in an (asymptotically free) Yang-Mills theory, IR renormalons would correspond to singularities in the semi-positive axis of the Borel plane, and render a Borel-Laplace resummation ambiguous. In this case, Maiezza and Vasquez argue that such IR renormalons can be traced back to Haag's theorem.}

Attempts to generate a finite result for the whole series through analytic continuation methods (such as using a Borel-Laplace transformation) necessarily rely on a choice of arbitrary constants. Maiezza and Vasquez describe these dependencies on an arbitrary choice as ``renormalization ambiguities.'' Thus, because of these ambiguities, ultimately arising from Haag's theorem, the interaction picture with perturbative renormalization cannot lead to their desired non-perturbative QFT.

\paragraph{Repair.}
According to Maiezza and Vasquez, there is no straightforward way to repair canonical QFT to address this problem. There remains a concerning flaw in perturbative renormalization, namely, the renormalized series' dependence on an arbitrary choice of constant. As Maiezza and Vasquez put it, perturbative renormalization ``cannot be a self-consistent cure, because perturbative renormalization needs to be completed, or in practice resummed'' \cite[page 10]{maiezza2020haag}.

In particular, they argue that the repair suggested by Klaczynski---essentially, to replace the interaction picture's unitary intertwiner with a non-unitary one---cannot work. The Dyson operator relating the free and interacting fields is unitary by construction, in such a way that simply undoing the time-ordered product present in its definition (as suggested by \cite{klaczynski2016haags}) does not therefore make it non-unitary \cite[5]{maiezza2020haag}. \footnote{Maiezza and Vasquez  point out that it can be written in a manifestly unitary form (see \cite[p. 80]{maiezza2020haag}). More generally, as a quantum mechanical evolution operator, the Dyson operator must be unitary, otherwise it would not preserve the normalization of quantum states. }

Something more is needed to repair the interaction picture, in particular to make it suitable for calculations using non-perturbative QFT.


With that said, Maiezza and Vasquez are quick to note that the practical uses made of the interaction picture with perturbative renormalization need no repair: they evade Haag's theorem. Haag's theorem is, after all, a non-perturbative result, in that it applies to the \textit{entire} power series expansion, not just the perturbative expansion at some finite order. Because the power series expansion is asymptotic, truncating at finite orders can give meaningful approximations, while also missing out on the non-perturbative effects that lead to the complete series' divergence. Thus, while regularization and renormalization do not `fix' the problem posed by the divergences tracked by Haag's theorem, this is essentially because Haag's theorem does not apply in cases where the series is truncated. Maiezza and Vasquez hold that perturbative renormalization fixes the divergences of individual terms within the series, but it does not address the non-perturbative divergences of Haag's theorem. \footnote{\cite{miller2018haag} and Maiezza and Vasquez agree about how perturbative renormalization works to fix the divergences of individual terms; however, Miller does not connect this to the overall divergence of the perturbative series)}.




\paragraph{Renovation.} They suggest three forms of renovation. First, in the short run, research efforts should aim at improving our understanding of the non-perturbative singularities that arise, such as renormalons. Further research should be directed into mathematical techniques to better understand neglected non-perturbative effects (for example \cite{Maiezza_resurgence, Maiezza_nonwilsonian}). These include resurgence, a method for reading off certain non-perturbative effects from the perturbative expansions (for a review, see \cite{Dunne:nonperturb}). 

Second, further work is needed work \emph{within} canonical QFT to account for non-perturbative effects. Ultimately, new physical insights are needed to complete quantum field theory. These must move beyond simply reformulating the perturbative quantum field theory framework [Maiezza and Vasquez, personal communication, Sept. 8, 2021].  

Third, Maiezza and Vasquez suggest physicists might need to consider more radical models, perhaps rejecting some of the standard assumptions of canonical QFT. They suggest one such possible model, dimensionally-reduced QFT (DRQFT) \cite{maiezza2023DRQFT}, in which the spacetime dimensions are dependent on the energy scale. Such an approach requires the introduction of a new mass-energy scale to encode some further unknown physics at high energy scales. \footnote{As such, this is not intended as a candidate for something like a ``fundamental theory''. Motivated by approaches in quantum gravity, they interpret the scale as encoding the effects of a classical field of geometric origins that weights long and short contributions differently. So the sense in which unknown physics is encoded in this model is quite different from the approach in Wilsonian renormalization.} However, something analogous to the interaction picture can be salvaged in such a theory. They demonstrate that such a theory can be made finite in perturbation theory, free of UV renormalon ambiguities. The key is that such a theory avoids one of the assumptions of Haag's theorem: the vacuum states of DRQFT are not Poincar{\'e}-invariant (see \cite[pages 17-18]{maiezza2023DRQFT}).



\subsection{Kastner}\label{subsec:Kastner}

\paragraph{Extra-Haagian outlooks.}
\cite{kastner_directaction}'s outlook differs from the previous authors. Insofar as she seeks the correct formulation of our theory of quantum phenomena, she sits with the physicists above. However, she thinks this correct formulation requires a fundamental reinterpretation of quantum phenomena. Here she goes further than the theory interpreters, who focus on new interpretations of existing physics.
Kastner advocates replacing QFT with a direct-action theory (DAT). As such, she has \textit{already} given up on mediating interactions locally in favor of nonlocal, direct interactions between field sources. This is a profound departure from QFT.
 



\paragraph{Assessment.} Kastner locates the problem further back conceptually. Rather, Kastner begins with the commitment that QFT is wrong from the outset. As such, her assessment of QFT in the light of Haag's theorem is grim: ``QFT is not the correct model; a different, yet empirically equivalent, model is needed'' (57). The problem lies with the foundational assumption of QFT\textemdash namely, that we use fields as mediators between particles (crudely, as ``a `bucket brigade' that is invoked in order to restrict causal influence to a local, continuous conveyance from spacetime point to spacetime point'' (59)). 
Thus, Haag's theorem ``simply tells us what we already know: the interaction picture of quantized fields does not really exist'' \cite[58]{kastner_directaction}. Haag's theorem is therefore an additional motivation for abandoning QFT in favor of DATs. 

\paragraph{Repair.} The treatment for such a deep problem is a major change in the modeling procedure for this area of physics. If we accept that Haag's theorem shows us that it simply does not work to model interactions by field operators creating and destroying Fock space states, then, Kastner urges, we are to instead revive direct-action theories (DAT), in which we model interactions through a direct, nonlocal interaction between sources of the field \cite[57]{kastner_directaction}.

Rather than modeling the evolution from free to interacting states (as in the interaction picture of QFT), on the direct-action approach, we do away with fields as mediators of interactions. Instead, we posit a direct, nonlocal  between particles or sources of the field, without needing any intermediary field, so that the issue of constructing a unitary operator to map between free and interacting theories does not arise. More fundamentally, in a DAT, quantum mechanics is only incorporated at the level of the sources and receivers of a field, and is not enforced through the imposition of canonical commutation relations on the fields (see \cite{davies1970quantum}). As such, there is no requirement of local commutativity between the field operators at spacelike separations, unlike in ordinary QFT (recall axiom \ref{axiom:localcommutativity}).


Kastner believes that this is a viable repair because she is convinced that DATs\textemdash which differ radically from QFTs\textemdash are empirically equivalent to QFTs. This places significant weight on (purported) demonstrations of the equivalence of direct-action theories and QFT (e.g., \cite{Narlikar1968}). If the demonstrations are sound, \textit{and} they extend such that all currently used QFTs are approximations to some DAT, then they would go some way towards explaining how calculations in renormalized quantum field theories are able to generate successful predictions: QFT's ``mathematical inconsistencies can be rendered inconsequential since they can be understood as arising from its ``makeshift,' nonfundamental character'' [63].\footnote{However, this is a non-trivial \textit{if}, especially given the recent results suggesting that empirical equivalence is a less straightforward concept than philosophers often presume. See \cite{weatherall_equiv1} and references therein for more detail.} 

Kastner further invokes the abandonment of the idea of the ether due to the Michelson-Morley experiment to justify abandoning QFT, saying that abandoning interactions through local quantum fields in response to Haag's theorem would have an analogous ``interpretive elegance'' (58).  She likens Haag's theorem to the EPR experiment insofar as it presents a serious challenge to the assumption that the laws of nature are entirely local. \footnote{Recall that Haag's theorem depends on an axiom of local commutativity (axiom \ref{axiom:localcommutativity}). Kastner's approach can evade Haag precisely because of this non-local character, avoiding such a local commutativity requirement.}
We repair all these problems, including those raised by Haag's theorem, by abandoning locality in favor of DATs. 




\paragraph{Renovation.} The prescribed renovation, on Kastner's view, is to develop a research program for DATs. This research program includes the rejection of alternative responses to Haag's theorem such as Haag-Ruelle scattering and constructive QFT; Kastner views these workarounds as ``\textit{ad hoc}, approximate, or partial measures'' (63). Future work in the research program would need to assess the question of the calculational, pragmatic viability of DATs as well as the crucial need to check thoroughly for empirical equivalence with all of the accepted theoretical results of ordinary QFT. Kastner cites work from Narlikar and Davies developing a DAT empirically equivalent to QED; more work is needed to determine if other sectors of the Standard Model of particle physics can be recovered in the DAT approach. Further rennovation efforts should also explore the possibility that a DAT approach would lead to novel, empirically testable predictions, as much of theoretical particle physics is currently disappointed with the lack of direction from particle accelerator experiments for new lines of research.

Given the wide extent of Kastner's proposed work\textemdash not least, replacing the entire Standard Model with DATs\textemdash it would perhaps be more appropriate to describe the proposal as an entirely new foundation than as a renovation of the one that exists.




\subsection{Seidewitz}\label{subsec:Seidewitz}


\paragraph{Extra-Haagian outlook.} Seidewitz's outlook is interesting in that it substantially coincides with several of the others already discussed, while at the same time departing from them. First, he seems to think of theory interpretation in a way similar to Earman and D. Fraser and the theory interpreters. Like the latter, he appears to take mathematical rigor as a precondition for providing a satisfactory interpretation \cite[369]{seidewitz2017avoiding}. However, Seidewitz sees axiomatic QFT built on the Wightman axioms as an \textit{incorrect} marriage of special relativity with quantum theory. In the light of this, Seidewitz thinks that these axioms will need to face significant revision. \footnote{See \cite[Section 5]{FraserDoreen2006Thesis} for an extended discussion of the sense in which Fraser takes the axioms of QFT to be provisional. }
Seidewitz departs from the puzzle-focused interpreters insofar as a univocal characterization of particle-ness is not forthcoming. This is because particles now come in two varieties\textemdash virtual and real\textemdash and a new definition of `particle' will have to account for the significant mathematical distinction between the two.\footnote{In particular, we should expect (an analog of) the (global) number operator characterization to hold only for real particles, given that the space-like evolution possible for virtual particles prevents their characterization as Fock space states to be created or destroyed (where this is done by operators used to define number operators).} 

Second, Seidewitz's outlook coincides with Kastner's. On the one hand, processes can evolve in a space-like way. That is, physics is non-local. Specifically, virtual processes need not evolve in a local fashion. On the other hand, particles come in two varieties\textemdash virtual and real. Indeed, the two agree on how these varieties are constituted with respect to locality. However, Seidewitz disagrees with Kastner that the non-locality of physics is a reason to abandon field theories. Rather, he suggests an alternative formulation of quantum field theory.


\paragraph{Assessment.} Thus, \cite{seidewitz2017avoiding}'s assessment of Haag's theorem is that it is the symptom of a larger problem, namely the inconsistent treatment of time in traditional QFT. Seidewitz reads Haag's theorem as a corollary of two previous results, where these previous results are the more significant. First, let $\hat{\psi}_{1}(t, \mathbf{x})$ and $\hat{\psi}_{2}(t, \mathbf{x})$ be field operators defined in Hilbert spaces $\mathcal{H}_{1}$ and $\mathcal{H}_{2}$, respectively, and suppose there exists a unitary operator $\hat{G}$ such that, at a specific time $t$, $\hat{\psi}_{2}(t, \mathbf{x})=\hat{G}\hat{\psi}_{1}(t, \mathbf{x})\hat{G^{-1}}$. Then (Theorem 1) the equal-time vacuum expectation values of the fields at time $t$ are the same. Second (Theorem 2), a given field's two-point expectation values satisfy a certain equation (Eq. (3) in Seidewitz) if, and only if, that field is free. Letting $\hat{\psi}_{1}$ be a free field and $\hat{\psi}_{2}$ a field related to the former as in Theorem 1, Haag's theorem merely observes that $\hat{\psi}_{2}$ also satisfies Theorem 2's equation for spacetime points at time $t$ (by Theorem 1) and, by the Lorentz-covariance of $\hat{\psi}_{2}$ and analytic continuation, extends $\hat{\psi}_{2}$'s satisfaction of the equation to any two positions.

As Seidewitz sees it, Haag's theorem ``essentially relies on a conflict between the presumption that the fields are Lorentz-covariant and the special identification of time in the assumptions of Theorem 1'' \cite[360]{seidewitz2017avoiding}. This special identification is apparent in that the time-evolution is given by the (frame-dependent) Hamiltonian. This requires particles to evolve on time-like trajectories. Essentially, the (frame-dependent) Hamiltonian operator is playing double duty as the generator of time translation, in addition to the generator of state evolution, with $t$ as the evolution parameter for each. However, because of this, the translation group of $t$---the unitary operators $\hat{G}(t)$---coincides with the group of Lorentz transformations, which guarantees that one can extend the coincidence $\hat{\psi}_{2}(t, \mathbf{x}) = \hat{G}\hat{\psi}_{1}(t, \mathbf{x})\hat{G^{-1}}$ at a specific time $t$ to \textit{all} times $t$.

\paragraph{Repair.} The immediate repair is to break the identification of the Hamiltonian with the generator of time translation and the energy observable. Seidewitz proposes we do so by dropping (i) the requirement that there is a unique Poincar{\'e}-invariant vacuum state and (ii) the spectral condition, i.e., the requirement that states are \textit{on shell}. Essentially, this involves dropping the assumption that states transform according to an \textit{irreducible} representation of the Poincar{\'e} group. This frees us up to define vacuum states as \textit{relative} to the choice of a Hamiltonian operator, where the Hamiltonian is no longer equated with the time-translation generator. Instead, a Hamiltonian is required only to be Hermitian, commute with all spacetime translations, and have a unique null eigenstate. \footnote{Note that this model is more general than standard AQFT. The Hamiltonian operators do not need to be positive definite or represent the energy observable. Likewise, each Hamiltonian will have a corresponding ``vacuum state'', but these need not correspond to the conventional sense of the ``ground state''.}

Then, one can consider the one-dimensional group generated by each Hamiltonian, where $\lambda$ (instead of $t$) is the (now frame-independent!) evolution parameter. (Physically, $\lambda$ can be thought of as the parameterization of a particle's path in spacetime---now unburdened from the constraint that the path be timelike.) Haag's theorem is evaded by limiting the equality of vacuum expectation values (Theorem 1) to equal values of $\lambda$, which is neither surprising---each Hamiltonian has a unique Poincar{\'e}-invariant vacuum state---nor problematic---one can no longer use Lorentz transformations to extend the equality any farther. In particular, one cannot demonstrate that the fields coincide (Haag's theorem). Thus, interacting fields need not coincide with free fields, and yet a version of the interaction picture is reinstated since interacting fields can be related to free fields by a unitary intertwiner \cite[369]{seidewitz2017avoiding}.



\paragraph{Renovation.} Like Kastner, the renovation Seidewitz proposes is far-reaching. Most important, it will involve expanding the reach of parameterized QFT. \footnote{In one sense, Seidewitz has formulated a more general framework than standard AQFT, one in which the vacuum state and Hamiltonian are both representation-dependent. In standard AQFT models, the energy operator is always the Hamiltonian and the evolution parameter is always time. Seidewitz' framework allows for additional models in which these alignments are broken.} As it stands, the theory does not cover gauge theories or non-Abelian interactions, nor does it resolve all of the problems with standard QFT. As such, significant work is required before this approach is of practical use. Nevertheless, insofar as it is conceptually closer to constructive QFT, one might reasonably suspect that the renovations required will be less thorough-going than those proposed by Kastner.




\subsection{Framework Results Summary}

Our framework accomplishes two main objectives. First, it maps out an otherwise wild jungle of scholarship on Haag's theorem. This organization is especially helpful as groundwork for making the interdisciplinary exchanges of ideas on Haag's theorem more efficient: it is not always easy for physicists, philosophers, and mathematicians to communicate. This is in part due to appropriately different aims, as well as differences in how these separate communities have developed their own forms of discourse, vocabulary, and standards of rigour. But at some point, the disparate aims of these communities reconnect, and facilitating interdisciplinary communication becomes essential to progress.

Thus, second, our framework goes beyond the organization of viewpoints on Haag's theorem by pulling each viewpoint's underlying disciplinary and methodological values to the foreground (see again table \ref{tab:EHO}). How one diagnoses and treats Haag's theorem is profoundly influenced by what one expects of QFT, and of the interaction picture in particular, as well as one's expectations of theoretical physics in general. One's expectations for the mathematician's and the philosopher's appropriate relationship to theoretical physics also has an influence.  


The viewpoints above are characterized by several different purposes. We schematically organize these purposes in figure \ref{fig:david}. Duncan and Miller are both concerned with a specific explanatory task: how, in fact, do physical predictions evade Haag's theorem? This takes for granted the acceptability of current predictions and appropriateness of QFT's scale-relative form. While Klaczynski makes similar observations, he moreover desires a QFT form that makes the algebraic structure of interaction dynamics mathematically manifest. Similarly, Maiezza and Vasquez hope to precisely characterize the non-perturbative structure of QFT. Klaczynski and Maiezza and Vasquez therefore coincide with Duncan and Miller in their attention to canonical QFT, but diverge in setting their sights on an unambiguous and non-perturbative structure underlying it. In contrast, Kastner and Seidewitz divert their attention from canonical QFT altogether in search of a more satisfactory marriage of special relativity and quantum theory.Finally, Earman and Fraser and the theory interpreters are doing the distinctly philosophical work of sorting out the downstream interpretive implications of Haag's theorem in conversation with the wider goals and conceptual foundations  of QFT.

\begin{table}[ptb]
    \centering
    \begin{tabularx}{\textwidth}{p{3.7cm}X}
        \toprule
        \textbf{Name} & \textbf{Extra-Haagian Outlooks} \\
        \midrule
        Earman \& D. Fraser & Which QM and SR assumptions characterize QFT? We desire well-defined representations of interactions, which requires carefully tracking our assumptions as is done in the formal variants of QFT. The process by which predictions are actually generated is irrelevant. \\
        \midrule
         \makecell{Theory Interpreters:\\ constrained puzzle}~~ & Should we interpret QFT interactions as the interaction of (free field) particles? Only the formal variants of QFT are ripe for interpretation. The calculational details by which predictions are actually generated is irrelevant. \\
        \midrule
        \makecell{Theory Interpreters:\\ beyond the puzzle}~~ & What is the particle content of particle physics? No restrictions are placed on representations, and the actual predictive process may be, and often is, relevant. \\
        \midrule
       Duncan \& Miller & Why are predictions in QFT reliable? It is to be shown that the actual process of generating predictions in canonical QFT, via renormalization and perturbation theory, is consistent. The details of the actual predictive process are not only relevant but central to this project. \\
        \midrule
        Klaczynski & What is the coherent structure underlying renormalized canonical QFT? We desire well-defined representations of interactions that can deliver unambiguous non-perturbative predictions. \\
        \midrule
        Maiezza \& Vasquez & To what extent can we recover an unambiguous non-perturbative formulation of QFT from the perturbative formulation of canonical QFT? We desire a well-defined representation of QFT that can address non-perturbative singularities. \\
        \midrule
       Kastner & How should QM be extended to the phenomena described by QFT? We desire a well-defined theory  whose predictions coincide with those of canonical QFT, but which eschews the CCRs. \\
        \midrule
        Seidewitz & How should SR be incorporated into QFT? We desire a well-defined theory that treats time analogously to the spatial coordinates. \\
        \bottomrule
    \end{tabularx}
    \caption{Extra-Haagian outlooks summary.}
    \label{tab:EHO}
\end{table}

\section{Where and how to make future progress}\label{sec:diagrams}



Our ultimate goal has been to provide guidance on QFT's construction in light of Haag's theorem. So what guidance does the framework provide? The reader might expect us to promote a QFT pluralism along the following lines. Section \ref{sec:lit} suggested there were substantial disagreements about how, or even whether, the QFT apparatus works. These disagreements appeared to concern the propriety of our present theory. Had the scrutiny of setion \ref{sec:framework} born this out, adjudication would have appeared necessary to ensure the security of our \textit{present} theory. But this was not the upshot of section \ref{sec:framework}. Instead, our framework revealed that respondents' purposes were not uniform. Moreover, the response-pluralism was decidedly future-directed in that the divergences were easier to appreciate when viewing the proposed renovations. In this way, different beliefs about how to pursue a \textit{future} theory may be driving the differing assessment's of Haag's theorem. We do not even know what this future theory is, let alone how we will discover it. Thus, humility would suggest leaving each of these researchers to their own pursuits.

We find this pluralistic response tempting, and we think there is much right in it. Nevertheless, we do not think it is sufficient for providing guidance on QFT's construction. In our opening vignette, we imagined researchers encountering Haag's theorem, and we asked how they should respond. While a reader could have, understandably, chalked it up as mere rhetoric, we take the vignette seriously: how should \textit{actual} researchers respond to the theorem? Actual researchers and actual research communities are subject to resource constraints, and they are thus forced to make choices. Indeed, some of these choices are \textit{difficult}, in the sense that there are few and fallible reasons to prefer one choice over another. Yet such choices are made, and reasons for these difficult choices are given. Consider that each respondent we survey not only pursues their own path, but moreover defends their decision to do so. While pluralism is right in that there appear to be multiple reasonable choices, it does not help us human decision-makers reason through the choices we actually face.

But we also do not think methodological monism\textemdash picking just a single path\textemdash is necessary for guidance at this stage. That is, we do not think the framework must force a choice upon researchers in order to provide guidance to them. Rather, we take it that any information that helps researchers make the difficult decision(s) of what path(s) to pursue constitutes guidance. Here we content ourselves with two kinds of guidance. In \ref{subsec:distinctlines}, we use the framework to separate the wheat of \ref{sec:lit} from the chaff by highlighting the very real choices of goal that underlie the more superficial disagreements. The above respondents have and provide reasons for these choices, and, given that the future of QFT is uncertain, these reasons are defeasible. Researchers are thus guided by the framework to choices of consequence to the future of physics and, \textit{ipso facto}, to reasoning of consequence.

In \ref{subsec:calltoaction}, we show how the framework, more generally, guides researchers from more traditional philosophical questions about our foundations to more methodological ones. Traditionally, philosophers have asked pointed questions about the current status of specific concepts or claims\textemdash questions about rigor and mathematical existence, about theory interpretation, about theory essence. But as section \ref{subsec:distinctlines} shows, these traditional questions morph into questions about what kinds of reasoning we can expect to be promising. We explain how the framework shifts our attention in this way, how these new questions differ from past debates about foundations, and sketch one path forward for a philosophy of physics focused on these new questions.


Again, our ultimate goal is to provide genuinely useful and usable guidance on the construction of QFT, however meager or provisional. The construction of QFT is the product of many choices and much reasoning about them. We emphasize that such reasoning is \textit{already} taking place in support of choices that we \textit{actually} face as we strive to bring greater specificity to our theoretical knowledge. This reasoning is consequential, whether it takes place in print, project selection, resource allocation, or post-conference debate. It is, therefore, reasoning worthy of attention.

\subsection{Three distinct lines of debate}\label{subsec:distinctlines}

\paragraph{Does the interaction picture exist?} 
The most salient of disagreements concerned the status of the interaction picture. On the one hand, Earman and D. Fraser, Klaczynski, and Maiezza and Vasquez each suggest its existence is undermined by Haag's theorem. On the other hand, Duncan and Miller both suggest it exists and its use is justified. The difference here is, of course, that Duncan and Miller speak of the interaction picture at intermediate stages in a perturbative calculation, where regularization ensures finitely many degrees of freedom. In those contexts, the interaction picture is perfectly well defined. Our framework was not necessary for this observation, nor was an appeal to extra-Haagian outlooks.

Yet there is more to say here, as revealed by the framework. First, the two groups differ substantially on what QFT should look like and, therefore, what counts as a sure foundation. Duncan and Miller are focused first and foremost on the reliability of the predictions generated by ``the algorithm of quantum field theory'' \cite[1135]{miller2021infrared}. While Miller clearly believes explicit identification of a non-perturbative structure is desirable \cite[818]{miller2018haag}, he nevertheless argues that the perturbative nature of our data does not undermine a realist interpretation of the theory \cite[1135]{miller2021infrared}. Thus, if the ultimate goal is to characterize a theory about which we can be realists, it would seem canonical QFT has a sufficiently sure foundation. Earman and D. Fraser, Klaczynski, and Maiezza and Vasquez do not seem similarly satisfied with canonical QFT's present perturbative structure. 

Second, members of the former group are not entirely alike in their accounts of the consequences of Haag's theorem. At issue is how deeply the authors trace Haag's theorem's consequence for the existence of a unitary intertwiner. For Earman and D. Fraser, the focus is on the existence of a unitary intertwiner between the free field and the \textit{unrenormalized} interacting field. In this case, an IR cutoff suffices to defuse Haag's theorem and ensure the existence of an unitary intertwiner. For Klaczynski and Maiezza and Vasquez, the focus is on the existence of a unitary intertwiner between the free field and the \textit{renormalized} interacting field. In this case, substantive engagement with renormalization techniques will be necessary for settling the question of whether a satisfactory intertwiner exists. Thus, applying our framework to this debate over the IP reveals that the real debate to be had is over the exact version of the fields for which we are in want of an intertwiner.

The framework therefore reveals two substantial choices underlying the superficial disagreement about the IP's existence. First, must we explicitly identify a non-perturbative structure in order to be satisfied with our theory's representation of scattering phenomena? Second, why, hence where, do we want a well-defined unitary intertwiner? In particular, is it enough for an intertwiner to relate the (IR safe and scale-relative) observables remaining after renormalization to appropriate free-field correlates (i.e., unitary in that amplitudes appropriately coincide), or do we want an intertwiner to find correlates for the entire free-field state space and observable structure in a representation of the interacting field (i.e., unitary in that it is a dynamical operator defined on a single state space)?


\paragraph{Is particle physics particle physics or not?}
Section \ref{subsec:particleproblem} also revealed disagreement about the theorem's consequences for a particle interpretation. While many theory interpreters, including Halvorson and Clifton, take it as evidence against a particle interpretation, Bain and Kastner, among others, do not. As we discussed above, the primary difference between (e.g.) Halvorson and Clifton and Bain is how wedded we are to the free-field conception of particles. Halvorson and Clifton are so wedded. Bain, instead, is not, for he aims to capture the sense in which particle physics \textit{is} about particles\textemdash if this is not the free-field one, so much the worse for the free-field conception.

The reader might think Kastner fits in with Bain. This is not correct, however. Kastner not only recommends a different interpretation of particle physics from Halvorson and Clifton but, as her proposed renovations reveal, suggests QFT must be replaced altogether. That is, rather than derive an interpretation from a presupposed formalism, she instead derives a formalism from a presupposed interpretation. Among other things, this interpretation is committed to a particle conception of particle physics. This commitment ultimately derives from a belief that relativistic quantum mechanics should behave more non-locally than QFT does, as, she thinks, non-relativistic quantum mechanics counsels. 
This suggests two disanalogies with Bain. First, the comparatively unmoored approximations of canonical QFT are not enough\textemdash we should demand exact calculations, or at least approximations to exact solutions. Second, a fully satisfactory theory should wear its metaphysics on its face more plainly than does canonical QFT.

As above, there are substantial choices underlying the superficial disagreement about particle existence. First, is the conception of particles inherited from the free-field formalism a satisfactory one for interacting QFT? Second, at what point should we be satisfied with a mathematical formalism in its representational capacity? In particular, is it (at least potentially) metaphysically significant that the calculations delivered by canonical QFT are not approximations to exact solutions?

\paragraph{What is relativistic QFT?} Spelling out the renovations proposed also reveals commitments we might have otherwise missed (see figure \ref{fig:disciplines}). Focusing just on broad strokes aims as a way to separate various responses to Haag's theorem, Kastner and Seidewitz would seem far removed from the puzzle-focused theory interpreters\textemdash the former propose new theoretical forms for QFT, while these interpreters restrict themselves to the formal variants. However, it is worth noting that these groups of authors seem to share the same outlook on axiomatization. In particular, both seem at least weakly committed to claiming that QFT phenomena has a structure that can be captured in a closed form as something like a finite set of axioms \textit{and that we can actually specify these axioms}. Likewise, the puzzle-focused interpreters who believe particles have been thoroughly banished believe not only that QFT can be so captured but that it has. In contrast, Earman and D. Fraser appear committed only to the belief that axiomatization could be \textit{helpful} for comprehending current QFT, or even for finding a more scrutable form. D. Fraser, at least, is quite clear that CQFT does not need to build a realistic model in order to have been useful\textemdash its utility lies rather in the fleshing-out of the relationship among the axioms. While Kastner and Seidewitz certainly could take such a view, the emphasis they place on the correctness of their chosen axioms suggests this is not their current conception.

This is useful, too, for understanding the varied responses to the question of which assumption of Haag's theorem should be dropped. Kastner, Seidewitz, and the most staid puzzle-focused interpreters each have a clear answer to this question. Things get murkier as we consider the others, however. Not only do none of them explicitly affirm some set of precise assumptions, they are not entirely clear about which assumptions are problematic and how. This comes in degrees, of course. Earman and D. Fraser are perhaps clearest, suggesting that the IP's assumption of a global unitary intertwiner is problematic but that the local equivalence delivered by Haag-Ruelle appears to suffice. Klaczynski and Maiezza and Vasquez are a bit less clear. Insofar as they each hope for an explicit formalism, unlike canonical QFT, they appear sympathetic to axiomatic approaches. 
 Klaczynski only suggests that renormalization \textit{somehow} replaces the unitarity of the intertwiner. 

 Maiezza and Vasquez are agnostic regarding the question of which assumptions leading to Haag's theorem should be dropped. In their proposed model DRQFT, they drop the assumption of a Poincar\'e-invariant vacuum state, but they present this as only one possible and promising path forwards.\footnote{Thus they propose relaxing the same assumption as Seidewitz. However, their model is expressed predominantly in the language of canonical, non-axiomatic QFT. This model is intended to encode unknown physics at high energy scales. [Maiezza and Vasquez, personal communication, Feb. 2, 2024] } 

In a sense, Duncan is even less clear, albeit for a principled reason. Obviously respectful of axiomatization's utility, he nevertheless emphasizes again and again that the truly remarkable feature of QFT is the ``scale separation'' property of local field theory \cite[570-1]{duncan2012conceptual}. Building one's theory around this fact, as the effective Lagrangian approach of canonical QFT has done, would seem to undermine specification of any single set of assumptions since, e.g., symmetry structure can vary. Duncan, it would seem, cannot be clearer about what is denied.

Here our framework, by distinguishing Assessments and Repairs from Renovations and Extra-Haagian Outlooks, calls our attention to a whole host of choices underlying the superficial disagreement. In general, these concern the utility and feasibility of formally-driven approaches to expanding QFT. In particular, what is the point of axiomatization? what can we reasonably expect it to afford us? how does (should) it interface with other programs in QFT?



\subsection{Lessons for making progress}\label{subsec:calltoaction}
In the first instance, the framework we have provided constitutes guidance for highlighting the choices driving the more superficial disagreements. But this is not its most important guidance. More dramatically, we think the framework highlights a kind of question that philosophers of physics often overlook regarding the so-called foundations.

Many are tempted to ask, as we seem to do in the title: how foundational is Haag's theorem to QFT, \textit{really}? The adverb signals a temptation to think there is a universal extra-Haagian outlook on QFT. Of course, as we see it, there is no such outlook. Yes, there is wide agreement that QFT long ago secured much comfortable space to wander around. But we have not rested there\textemdash we wish to expand these rooms, and in so doing we soon face choices about how best to secure our foundations. Many putative problems have arisen from attempts to do so, Haag's theorem among them. Is evading Haag's theorem central to \textit{all} of these expansions? That is, does Haag's theorem need to be explicitly centered in \textit{all} expansion programs\textemdash say, by setting its evasion as a derivational goal, or by cementing involved concepts as essential to the representation of QFT? Surely not\textemdash even Hilbert recognized that such designations were relative to the proposed expansion and, at any rate, only tentative \cite{hilbert-axiomatisches}. If an expansion does not call for, say, securing global unitary equivalence between free and interacting fields, or for Poincar{\'e} covariance of the full state space, or unambiguous analytic continuation to finite coupling, or some other such security, Haag's theorem simply is not relevant.

Nevertheless, there are expansions for which Haag's theorem is relevant, so one can easily ask the same question in a localized form: how foundational is Haag's theorem for \textit{this} expansion? This seems a more sensible question, and in fact it seems to be asking several sensible questions. Most straightforwardly, it could be asking how radically the form of QFT needs to change in order to secure the proposed  expansion against the ravages of Haag's theorem. The radicalness of Kastner's and Seidewitz's expansions stand out here, suggesting that Haag's theorem is profoundly foundational for their expansions; on the opposite end, Duncan and Miller propose very little in the way of expansion, and Haag's theorem therefore does not appear foundational since \textit{no} change is necessary. Call questions like this \textit{formal questions} about how foundational Haag's theorem is. This seems to us to be how most philosophers of physics would approach grading how foundational a (cluster of) results is: the more `central' a \textit{concept} or \textit{statement} is within the theoretical apparatus in some formal sense\textemdash cashed out either semantically or syntactically\textemdash the more foundational it is.


Yet it is more difficult to answer the formal question about how foundational Haag's theorem is with respect to the other proposed expansions. This is because these expansion projects are actively probing the theoretical and mathematical possibilities for achieving their goals and, as such, Haag's theorem does not have as definitive a formal status. But this reveals a set of \textit{methodological questions} you might have in mind when asking about how foundational Haag's theorem is. Each of the expansion projects takes a position on the role of reasoning done with more abstract or idealized assumptions, more exacting characterizations of representations, or, in general, at farther remove from any particular, concrete modeling situation. Call this foundational \textit{reasoning}, to distinguish it from the (putatively foundational) concepts or statements of the last paragraph. Note that the two senses of foundational need not coincide: ultimately Haag's theorem may not require much adjustment to address (formal sense), but clearly the reasoning that led to it, and that expanded around it, is foundational in this methodological sense.

Rather than inquire about the formal status of Haag's theorem, we could inquire  about reasoning of its ilk. How fruitful was the reasoning that led to Haag's theorem? Is such reasoning still useful, so that it remains deserving of attention? Why or why not? Such questions are \textit{methodological} insofar as they primarily concern the nature of the \textit{growth} of human knowledge rather than the nature of its currently supposed contents. In more general terms, we might ask:
\begin{enumerate}
        \item \textbf{What role does (should) foundational reasoning play in progress in physics?} 
        \item \textbf{How does (should) foundational reasoning coordinate with non-foundational work?}  
        \item \textbf{And, what does (should) foundational reasoning even look like?} 
\end{enumerate}
Asking these methodological questions opens up a comparatively new avenue of inquiry for philosophers of physics by shifting our attention from the past\textemdash what is the formal-foundational status of Haag's theorem\textemdash to the future\textemdash how fruitful will Haag's theorem-related reasoning be.

And here is the most significant contribution of the framework: in attempting to decide on the status of our QFT knowledge vis-{\'a}-vis Haag's theorem, we have been dragged into a decision context where we must forecast the future of the various proposals. In Nickle's terms \cite{Nickles2006-NICHAC}, we have been dragged from philosophers' usual concern for \textit{epistemic} appraisal\textemdash is it well-confirmed or truth-conducive?\textemdash to a concern for \textit{heuristic} appraisal\textemdash does it have a good ``opportunity profile''? Our positions is that, when it comes to thinking hard about our best physical theories, being dragged in this way is neither new nor unwelcome. Indeed, as Streater and Wightman humorously (but correctly) observe, many physicists felt that the development of a more rigorous \textit{quantum} field theory was supremely ill-advised; cynics had even ``compared them to the Shakers, a religious sect of New England who built solid barns and led celibate lives, a non-scientific equivalent of proving rigorous theorems and calculating no cross sections'' \cite[1]{streater2000pct}. That is, the ``opportunity profile'' of (rigorous) foundational reasoning about quantum field theory was considered weak.

The early negative assessment of rigorous approaches to QFT was born of beliefs about the comparatively greater promise of investigating the \textit{classical} foundations of field theory versus the foundations of a \textit{quantum} field theory. And yet, many of the results that give comfort in the face of Haag's theorem\textemdash Reed and Simon's Theorem X.46 and the contemporary understanding of LSZ among them\textemdash were at least informed by the quantum Shakers above who aimed to ``kill it or cure it'', as Streater and Wightman put it \cite{streater2000pct}. So, does this vindicate the kind of foundational reasoning Shaker-physicists, such as Haag or Wightman, engaged in? And if it does, does it tell us anything about the auspiciousness of such reasoning today? These are the kinds of methodological questions to which our attention has been drawn by the framework's centering of scientific reasoning rather than individual concepts or claims.


We think asking methodological questions about foundational \textit{reasoning} better equips philosophers for engaging with ongoing inquiry. After all, our respect of science does not boil down only to our trust of the propositional claims scientists make about what the world is like. Rather, we presumably trust these claims because we trust science is composed of experts ``knowing \textit{how} to conduct research, including HA [heuristic assessment] of the various proposals for fruitful inquiry'' \cite[84]{Nickles2015}. In particular, we trust that they have the expertise necessary to decide which theorems or concepts it would be fruitful to place at the foundation of a theory.

$\dots$Or, at least, we trust they have the expertise necessary to do so \textit{in the long run}. But, as Haag's theorem illustrates nicely, this is where it matters that we are forced to make choices: \textit{how} do we get to this ``in the long run''? Consider a small-scale version of this problem we face with Haag's theorem. Many folks agree that finding an unambiguous structure underlying canonical QFT's perturbative calculations would be significant progress, and some are willing to set this as a goal in itself. Insofar as finding this structure would seem to necessitate identifying the analytic structure to which perturbative calculations are approximations, the phenomena of Haag's theorem are directly implicated. This makes Haag's theorem relatively foundational in the formal sense. Nevertheless, this does \textit{not} mean that foundational reasoning like that which led to Haag's theorem is the most fruitful \textit{kind} of foundational reasoning. Indeed, it is even possible that identifying the structure we are after will come from experimental serendipity and not reasoning we presently recognize as foundational.

Nevertheless, serendipity is not an especially defensible \textit{plan} for pursuing future physics, which leaves us to decide among the kinds of foundational reasoning currently on offer. Supposing we have set ourselves the goal of specifying an unambiguous structure underlying canonical QFT's perturbative calculations, how should we go about doing so? As the pluralist response we began this section with emphasized, it seems likely that a mixed, rather than pure, strategy will be appropriate. We simply do not know definitively which program(s) will bear fruit. But the question is, \textit{which} mixed strategy? This is where it would be helpful to compare and contrast the extant programs, considering especially how they complement or compete against alternatives as well as the reasons given for or against their auspiciousness. Should these reasons be historical or abstracted away from the particulars of the problem of, say, Haag's theorem, there is all the more reason for philosophers of science to get involved.

It may seem that we are calling for a retread of the path worn by the so-called D. Fraser-Wallace debate \cite{Wallace2006_naivite,wallace2011taking,fraser2011take,fraser2009quantum} (see also \cite{koberinski2016reconciling,koberinski_haag} for reception of the debate). There is something to this charge. Both D. Fraser and Wallace have views on the auspiciousness of various proposed expansions of QFT, and later work from both comes to consider questions like those we pose above more directly (e.g., \cite{fraser2022justifying,Wallace2020framework}). However, the debate as it exists in print is ultimately concerned with what form of QFT is worth subjecting to philosophical methods of interpretation, as well as what those methods should be. It is not (explicitly) about the nature or role of foundational reasoning \textit{for physics}, let alone about the various goals physicists and mathematicians have set for themselves, the auspiciousness of the programs they champion, or the reasons they give for such assessments of auspiciousness. 

The present work is better recognized as calling for extending a path concerned with the scientific value of different kinds of foundational reasoning. \cite[14]{koberinski2023generalized}, for example, argues that the axiomatic-adjacent approach of framework generalizations is reasonable insofar as it ``expand[s] the domain of indirect tests able to find evidence of new physics.'' In another vein, \cite{MITSCH2022Hilbert-style} argues that no-go theorems can anchor mathematical research programs in serious interpretive issues. \cite{FGMhowto} argue, moreover, that no-go theorems can serve as helpful guides to building adequate models. We are suggesting that further work on the utility of the various conceptual apparatuses, methods, exploratory techniques, etc., on offer could benefit scientific reasoning about what projects to take up and, thereby, the choices made. Figure \ref{fig:disciplines} provides a loose disciplinary sort of some of the programs we encountered in our discussion of Haag's theorem; each is arguably engaged foundational reasoning, which makes them good targets for further investigation along this path. In doing so, we should keep two things in mind. First, QFT is not a methodological island: there is a rich history of physics\textemdash and indeed of other sciences!\textemdash on which we can draw as we adumbrate the virtues and vices of the various reasoning tools in our tool belt. And second, we must be ready to adapt when we inevitably face new goals, pragmatic constraints, or tool adaptations, and the like.

\vspace{1cm}

\tikzstyle{Question} = [rectangle, text width=15.5em, text centered, rounded corners, minimum height=1.5em, draw=black, fill=orange!30]
\tikzstyle{Question2} = [rectangle, text width=10.5em, text centered, rounded corners, minimum height=1.5em, draw=black, fill=orange!30]
\tikzstyle{Answer} = [rectangle, text width=15.5em, text centered, rounded corners, minimum height=1.5em, draw=black,fill=green!30]
\tikzstyle{Answer2} = [rectangle,  text width=5.5em, text centered, rounded corners, minimum height=1.5em, draw=black, fill=green!30]
\tikzstyle{Author} = [rectangle,  minimum width=3cm, minimum height=1cm, text centered, draw=black, fill=blue!30]

\pagebreak

\begin{figure}[htp]
\centering
\begin{tikzpicture}[thick,scale=1.3, every node/.style={scale=0.7}]

\node (Game) [Question] at (1, 0.5) {In most general sense, what ``game'' are we playing?};
\node (SRNRQM) [Answer] at (-2, -1) {Investigate the union of SR and NRQM.};
\node (besttheory) [Answer] at (4, -1) {Explain/assess/improve \\ theoretical physics in practice.};

\node (logstruc) [Question2] at (-4, -2.5) {What is the logical structure of the foundations of QFT?};
\node (ontol) [Question2] at (0, -2.5) {What is the  ontology of `particle' physics?};

\node (EF) [Author] at (-4, -3.5) {Earman \& D. Fraser};
\node (PI) [Author] at (0, -3.5) {Theory Interpreters};

\node (framework) [Question] at (4, -4) {Does Haag's theorem call for radical revisions to QFT?};
\node (new) [Answer2] at (-2, -6) {Yes};
\node (keep) [Answer2] at (4, -6) {No};

\node (S) [Author] at (-3, -7) {Seidewitz};
\node (K) [Author] at (-1, -7) {Kastner};

\node (renorm) [Question] at (2, -8) {How do current perturbative renormalization techniques circumvent Haag's theorem?};

\node (Poin) [Answer] at (6, -9.5) {Regularization breaks Poincar\'e invariance.};
\node (IPsaved) [Answer] at (6, -10.5) {Existing practice with IP is salvaged.};
\node (DM) [Author] at (6, -11.5) {Duncan \& Miller};

\node (Unit) [Answer] at (2, -9.5) {Renormalization entails a non-unitarity intertwiner.};
\node (IPbroken) [Answer] at (2, -10.5) {We need to replace the IP.};
\node (KI) [Author] at (2, -11.5) {Klaczynski};

\node (perturb) [Answer] at (-3, -9.5) {Do not specify. However, the problems of Haag's theorem re-emerge when we consider the whole perturbative expansion.};
\node (newid) [Answer] at (-3, -10.5) {Will require new physical insights.};
\node (MV) [Author] at (-3, -11.5) {Maiezza \& Vasquez};

\Edge[](Game)(SRNRQM)
\Edge[](Game)(besttheory)

\Edge[](SRNRQM)(logstruc)
\Edge[](SRNRQM)(ontol)

\Edge[](logstruc)(EF)
\Edge[](ontol)(PI)

\Edge[](besttheory)(framework)
\Edge[](framework)(new)
\Edge[](framework)(keep)
\Edge[](new)(S)
\Edge[](new)(K)
\Edge[](keep)(renorm)

\Edge[](renorm)(Poin)
\Edge[](renorm)(Unit)
\Edge[](renorm)(perturb)

\Edge[](Poin)(IPsaved)
\Edge[](IPsaved)(DM)

\Edge[](Unit)(IPbroken)
\Edge[](IPbroken)(KI)

\Edge[](perturb)(newid)
\Edge[](newid)(MV)
\end{tikzpicture}
\caption{A mapping of the responses to Haag's theorem according to the driving motivations.}
\label{fig:david}
\end{figure}
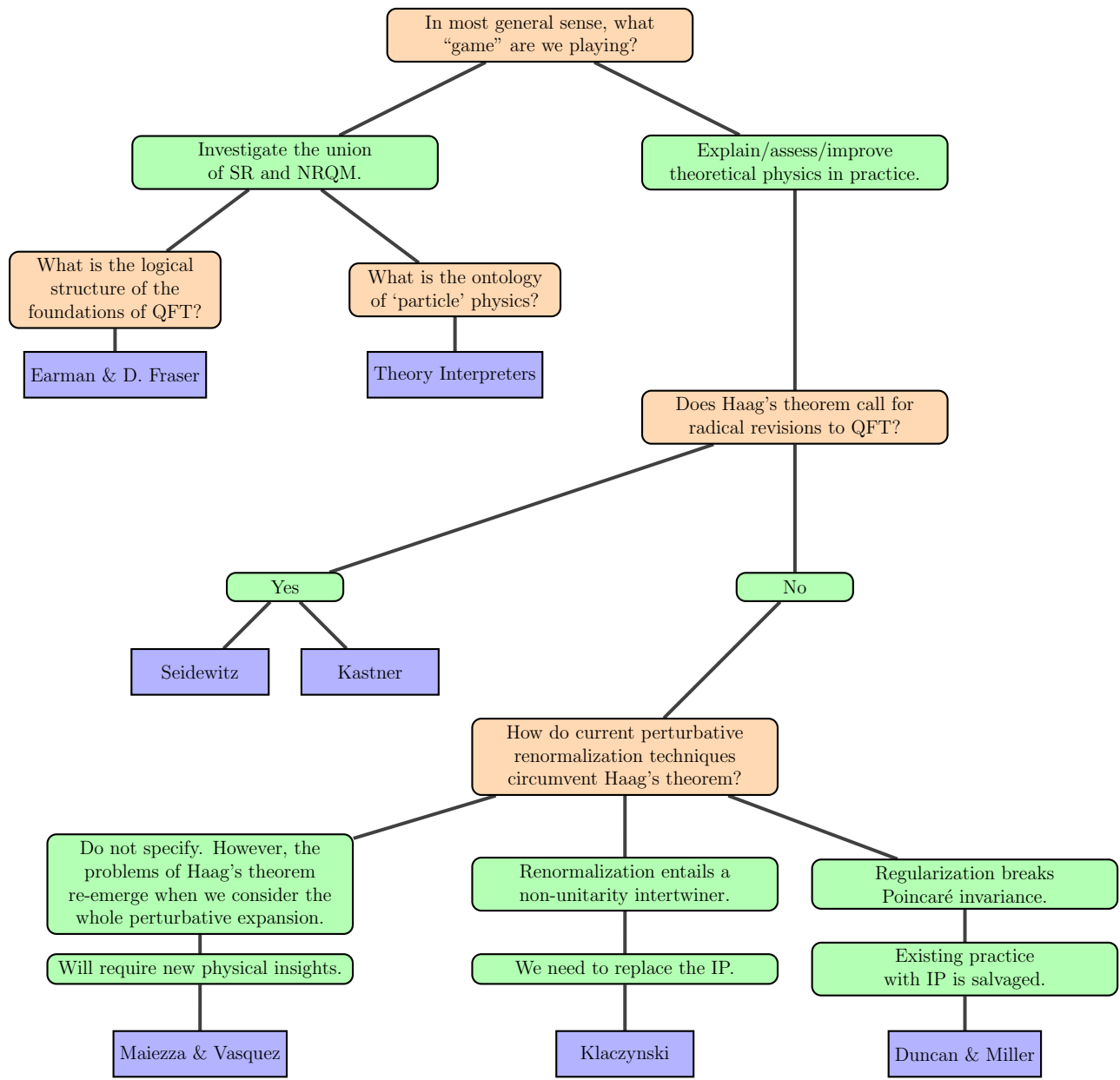

\pagebreak

\tikzstyle{MathAnswer} = [circle,  minimum width=3cm, text centered, draw=black, fill=blue!30]
\tikzstyle{ExpAnswer} = [circle,  minimum width=3cm, text centered, draw=black, fill=green!30]
\tikzstyle{PhilAnswer} = [circle,  minimum width=3cm, text centered, draw=black, fill=red!30]

\tikzstyle{MathArea} = [rectangle, minimum width=3cm, text centered, draw=black, fill=blue!20]
\tikzstyle{ExpArea} = [rectangle, minimum width=3cm, text centered, draw=black, fill=green!20]
\tikzstyle{PhilArea} = [rectangle, minimum width=3cm, text centered, draw=black, fill=red!20]

\tikzstyle{EnergyAuthor} = [rectangle, rounded corners, minimum width=3cm, minimum height=1cm, text centered, draw=black, fill=yellow!10]

\tikzstyle{EnergyAuthor} = [rectangle, rounded corners, minimum width=3cm, minimum height=1cm, text centered, draw=black, fill=yellow!10]

\begin{figure}[!htbp]
\centering
\begin{tikzpicture}[thick,scale=1.18, every node/.style={scale=0.7}]

\node (Energy) [Question] at (0,0.5) {What kind of work do the repairs, maintenance and/or renovations call for?};

\node (Math) [MathAnswer] at (-3.75, -1.5) {Mathematics};
\node (Exp) [ExpAnswer] at (0, -6) {Physics};
\node (Phil) [PhilAnswer] at (3.75, -1.5) {Philosophy};

\Edge[](Energy)(Math)
\Edge[](Energy)(Exp)
\Edge[](Energy)(Phil)

\node (Renorm) [MathArea] at (-2, -3) {Renormalization};
\node (Models) [MathArea] at (-6, -3) {Models};

\Edge[](Math)(Renorm)
\Edge[](Math)(Models)

\node (EF1) [EnergyAuthor] at (-6, -5.25) {Earman \& D. Fraser};
\node (Seid) [EnergyAuthor] at (-5, -4.25) {Seidewitz};

\node (Klac) [EnergyAuthor] at (-1.5, -4.25) {Klaczynski};
\node (MV1) [EnergyAuthor] at (-3.5, -5.25) {Maiezza \& Vasquez};

\Edge[](Renorm)(Klac)
\Edge[](Renorm)(MV1)
\Edge[](Models)(EF1)
\Edge[](Models)(Seid)

\node (Betas) [ExpArea] at (-4.5, -8) {Effective Field Theory};
\node (Insight) [ExpArea] at (1.5, -8) {Physical Insight};
\node (Cutoff) [ExpArea] at (-1.5, -8) {Cutoff Invariance};
\node (Toolkit) [ExpArea] at (4.5, -8) {New Theoretical Commitments};

\Edge[](Exp)(Betas)
\Edge[](Exp)(Cutoff)
\Edge[](Exp)(Insight)
\Edge[](Exp)(Toolkit)

\node (DM) [EnergyAuthor] at (-4.5, -9) {Duncan \& Miller};
\node (Dunc) [EnergyAuthor] at (-1.5, -9) {Duncan};
\node (MV2) [EnergyAuthor] at (1.5, -9) {Maiezza \& Vasquez};
\node (Kast) [EnergyAuthor] at (4.5, -9) {Kastner};

\Edge[](Betas)(DM)
\Edge[](Cutoff)(Dunc)
\Edge[](Insight)(MV2)
\Edge[](Toolkit)(Kast)

\node (Interpretation) [PhilArea] at (2, -3) {Traditional Interpretation};
\node (InterpMeth) [PhilArea] at (6, -3) {New Interpretive Methods};

\Edge[](Phil)(Interpretation)
\Edge[](Phil)(InterpMeth)

\node (EF2) [EnergyAuthor] at (1.5, -4.25) {Earman \& D. Fraser};
\node (HC) [EnergyAuthor] at (3.5, -5.25) {Halvorson \& Clifton};

\node (Wall) [EnergyAuthor] at (6, -5.25) {Wallace};
\node (Bain) [EnergyAuthor] at (5, -4.25) {Bain};

\Edge[](Interpretation)(EF2)
\Edge[](InterpMeth)(Wall)
\Edge[](Interpretation)(HC)
\Edge[](InterpMeth)(Bain)




\end{tikzpicture}
\caption{Disciplinary homes of proposed renovations projects.}
\label{fig:disciplines}
\end{figure}
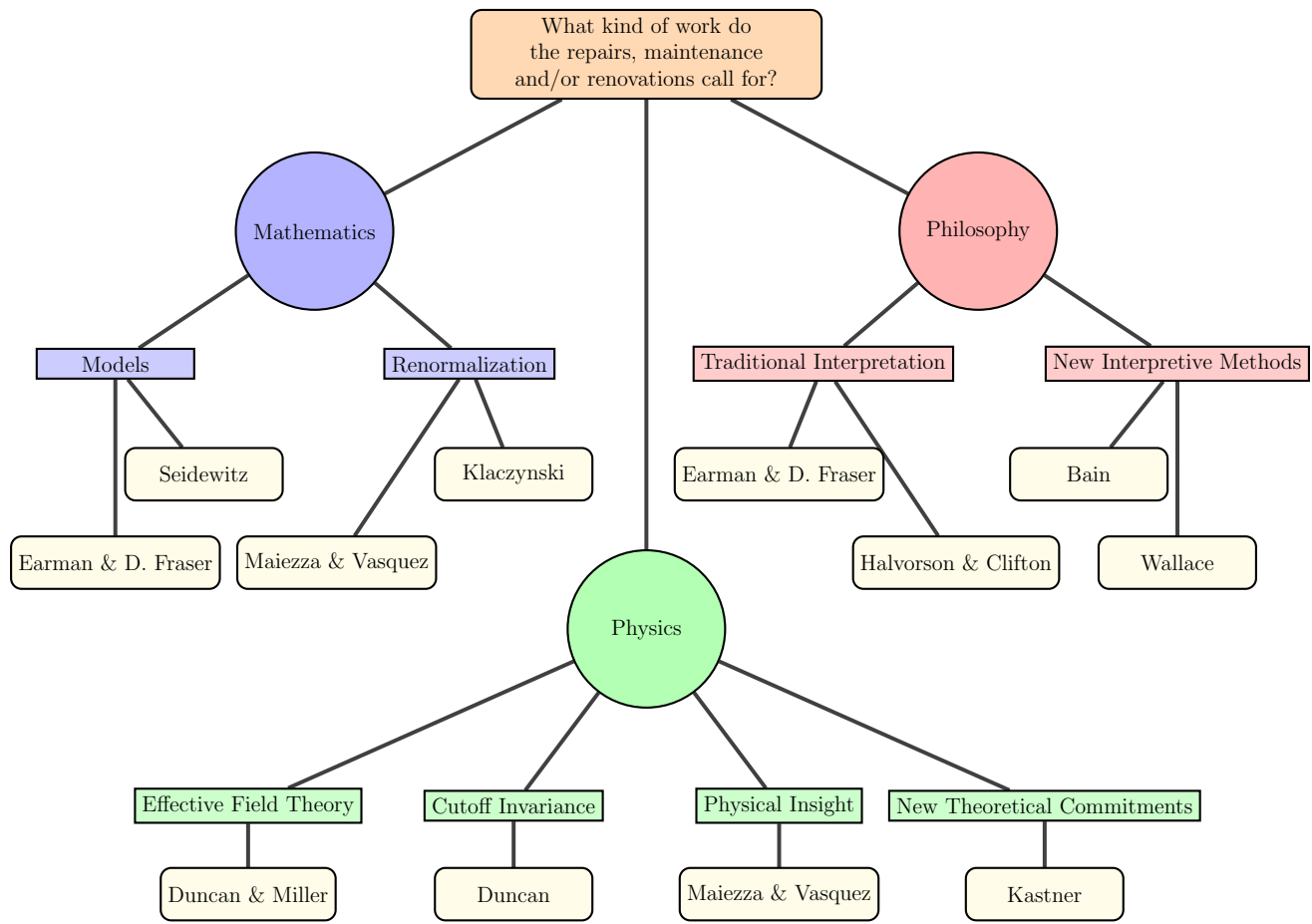

\section{Conclusion}\label{sec:conclusion}
Haag's theorem cries out for explanation and critical assessment: it sounds the alarm that something is (perhaps) not right in how QFT has been built. 
Viewpoints as to the precise nature of the problem (assessment), the appropriate solution (repair), and subsequently called-for developments in areas of physics, mathematics, and philosophy (maintenance or renovation) differ widely. Moreover, the extant literature presenting these differing views constitutes a complex  mix of arguments at cross-purposes, generating substantive confusion as to the precise issues to be addressed. In this paper, we have worked to address this confusion by cataloging and comparing a number of these viewpoints. We have developed and then deployed a framework for accomplishing that task. The application of our framework reveals each authors' background disciplinary and methodological commitments and expectations of QFT\textemdash what we have termed extra-Haagian outlooks. 

We have argued that these extra-Haagian outlooks are the primary driving forces in the various debates on Haag's theorem, and that these furthermore shape the different viewpoints as to what are the most important future projects regarding QFT, as well as on the appropriate methods for accomplishing those projects. We have urged stakeholders in the future of QFT to reconsider their own expectations for QFT, reflecting seriously on meta-level questions regarding the nature of foundational work in physics. What role does (should) foundational work play in the progress in physics? How does (should) foundational work coordinate with non-foundational work? What does (should) foundational work even look like? In calling attention to these questions, we encourage stakeholders\textemdash especially philosophers\textemdash to assess the significance of the reasoning surrounding theorems like Haag's rather than the theorems themselves. This, we believe, promises to deliver more useful guidance by engaging with the reasoning physicists and mathematicians are already doing as they make difficult choices about how to proceed with the construction of QFT.



\vspace{2cm}


%


\pagebreak

\renewcommand\refname{Bibliography}
\bibliographystyle{apalike}
\bibliography{HaagBibliography.bib}{}

\end{document}